\documentclass[pra,twocolumn,preprintnumbers,amsmath,amssymb,superscriptaddress,longbibliography]{revtex4-2}

\usepackage{color}
\usepackage{graphicx}% Incl fgfrude figure files
\usepackage{dcolumn}% Align table columns on decimal point
\usepackage{bm}% bold math
\usepackage{hyperref}
\usepackage{enumitem}
\usepackage{balance}
\usepackage{comment}
\usepackage{cleveref}

\usepackage{float}
\usepackage{multirow}
\usepackage{makecell}

\usepackage{amsthm}
\usepackage{svg}

\usepackage{xcolor}   
\usepackage{physics}

\theoremstyle{plain} %italicizes text
\newtheorem{theorem}{Theorem}[section]

\begin{document}

\title{Preparing squeezed, cat and GKP states with parity measurements}

\author{Zhiyuan Lin}
\thanks{The indicated authors are joint first authors}
 \affiliation{State Key Laboratory of Precision Spectroscopy, School of Physical and Material Sciences, East China Normal University, Shanghai 200062, China}

\author{Sen Li}
\thanks{The indicated authors are joint first authors}
\affiliation{State Key Laboratory of Precision Spectroscopy, School of Physical and Material Sciences, East China Normal University, Shanghai 200062, China}

\author{Jingyan Feng}
\affiliation{New York University Shanghai; NYU-ECNU Institute of Physics at NYU Shanghai, 567 West Yangsi Road, Shanghai, 200124, China.}

\author{Kaixuan Zhou}
\affiliation{New York University Shanghai; NYU-ECNU Institute of Physics at NYU Shanghai, 567 West Yangsi Road, Shanghai, 200124, China.}
\affiliation{Department of Physics, New York University, New York, NY 10003, USA}

\author{Theodore Mollano}
\affiliation{New York University Shanghai; NYU-ECNU Institute of Physics at NYU Shanghai, 567 West Yangsi Road, Shanghai, 200124, China.}
\affiliation{Williams College, Williamstown MA 01267, USA.}

\author{Valentin Ivannikov}
\affiliation{New York University Shanghai; NYU-ECNU Institute of Physics at NYU Shanghai, 567 West Yangsi Road, Shanghai, 200124, China.}

\author{Matteo Fadel}
\email{fadelm@phys.ethz.ch}
\affiliation{Department of Physics, ETH Zürich, 8093 Zürich, Switzerland}

\author{Tim Byrnes}
\email{tim.byrnes@nyu.edu}
\affiliation{New York University Shanghai; NYU-ECNU Institute of Physics at NYU Shanghai, 567 West Yangsi Road, Shanghai, 200124, China.}
\affiliation{State Key Laboratory of Precision Spectroscopy, School of Physical and Material Sciences, East China Normal University, Shanghai 200062, China}
\affiliation{Center for Quantum and Topological Systems (CQTS), NYUAD Research Institute, New York University Abu Dhabi, UAE.}
\affiliation{Department of Physics, New York University, New York, NY 10003, USA}

\begin{abstract}
Bosonic modes constitute a central resource in a wide range of quantum technologies, providing long-lived degrees of freedom for the storage, processing, and transduction of quantum information. 
Such modes naturally arise in platforms including circuit quantum electrodynamics, quantum acoustodynamics, and trapped-ion systems. 
In these architectures, coherent control and high-fidelity readout of the bosonic degrees of freedom are achieved via coupling to an auxiliary qubit. 
When operated in the strong dispersive regime, this interaction enables parity measurements of the mode which, in combination with phase-space displacements, constitute a standard experimental tool for full Wigner-function tomography.
Here, we propose a protocol based on displaced parity measurements that allows for the preparation of a variety of bosonic quantum states.
We demonstrate the generation of squeezed states, achieving $\sim 9\,$dB of quantum noise reduction after three parity measurements, and larger squeezing with an increasing number of measurements in the lossless case. The technique can be generalized to the preparation of other paradigmatic bosonic states, including cat and Gottesman–Kitaev–Preskill states. Using more general dispersive measurements and displacements, we show that the scheme is universal, such that it is possible to prepare an arbitrary state. 

\end{abstract}

\date{\today}

%\pacs{}
\maketitle

\section{Introduction}
\label{sec:introduction}

Heisenberg's uncertainty principle states that two canonically conjugated observables (i.e. position and momentum) have limits on the precision of their simultaneous estimate.
This is formalized by stating that the product of their variances cannot be smaller than a value set by quantum mechanics. 
While particular types of states (e.g. coherent states) have equal uncertainties in the two observables, squeezed states reduce the uncertainty in one variable at the expenses of increasing the uncertainty in the other.  
This allows for a way of reducing quantum noise below the standard quantum limit, where the noise of a measurement is limited by quantum noise. 

Since the first observation of squeezing \cite{breitenbach1997measurement}, experimental generation is now commonplace with best squeezing levels up to 15 dB in optical systems \cite{vahlbruch2016detection,schonbeck201713}. 
Methods of generating squeezed states include parametric processes from nonlinearities  \cite{slusher1985observation,banaee2008squeezed,Stefszky_2011,nadgaran2023squeezed,eichler2011observation,yurke1989observation}, conditional preparation based on quantum nondemolition measurements  \cite{colangelo2017simultaneous,hosten2016measurement,chen2011conditional,schleier2010states},  coupling to additional modes \cite{wineland1992spin,aspelmeyer2014cavity,leroux2010implementation}, nonadiabatic techniques \cite{xin2023long}, and many-body dynamics \cite{berrada2013integrated,strobel2014fisher,lucke2014detecting}. 
The primary application of squeezed states is for quantum metrology, where they may be used for precision measurements \cite{giovannetti2011advances,pezze2018quantum,byrnes2021quantum,jabir2024quantum}, but they also play a crucial role in quantum communication \cite{braunstein2005quantum} and simulation \cite{LundPRL14,HamiltonPRL17}. 

While optical systems are canonical systems for generating squeezed states, several other platforms have demonstrated highly-controllable bosonic modes where squeezed and other quantum states can be engineered.  
Examples include microwave cavity modes in circuit quantum electrodynamics (cQED) \cite{blais2021circuit,li2011engineering,krasnok2024superconducting}, vibration modes in circuit quantum acoustodynamics (cQAD) \cite{marti2024quantum}, or motional states of trapped ions \cite{wolf2019motional,jia2022determination,sutherland2021motional}. 
Interestingly, in these systems the bosonic mode couples to a discrete degree-of-freedom, such as a qubit, through a coupling often described by the Jaynes-Cummings model \cite{blais2021circuit,LeibfriedRMP03}. 
The qubit can be used as the nonlinear element required to engineer effective squeezing dynamics for the bosonic mode. 
Specifically, squeezing is generally realized by parametric processes where the system is driven off-resonantly at suitable frequencies.
The best squeezing that has been attained for these system are 8 dB for cQED {\color{red} \cite{dassonneville2021dissipative,pan2023protecting},} around 3 dB for cQAD \cite{marti2024quantum}, and 5 dB in trapped ions \cite{sutherland2021motional}.

Squeezed states are also the building blocks for other types of states which are useful in contexts beyond quantum metrology. 
A paradigmatic example is the Gottesmann-Knill-Preskill (GKP) code \cite{gottesman2001encoding}, where the logical states consist of superpositions of equally displaced squeezed states.  This allows for a method of forming an error protected qubit by exploiting the infinite-dimensional Hilbert space of a bosonic mode \cite{brady2024advances}. 
Although ideal GKP codewords have infinite energy, as they involve infinite squeezing, finite-resources approximations are nevertheless extremely useful for a variety of applications such as quantum repeaters \cite{rozpkedek2021quantum}, quantum metrology \cite{valahu2025quantum}, and error correction \cite{royer2022encoding}. 
Approximated GKP states have been prepared on a variety of platforms, such as trapped ions \cite{fluhmann2019encoding}, superconducting microwave cavities \cite{campagne2020quantum}, and photonics \cite{konno2024logical,larsen2025integrated}.  The primary method that is used to generate such states is ``breeding'' \cite{weigand2018generating}, where conditional displacements are performed, then interfered. 
While great efforts have already been devoted to its generation, the complex structure of GKP state makes challenging to generate such states, and more efficient or complementary methods are always desirable.

In this paper, we introduce methods for generating squeezed, cat, and GKP states based on displacement operations and parity measurements.  More generally, we consider the combination of dispersive measurements and displacement operations as a state generation primitive.  As such we call the general framework the Parity/Number basis measurement Displacement Algorithm (PANDA). Our scheme using parity measurements is illustrated in Fig.~\ref{fig1}.  By combining parity measurements and displacements we show that a variety of states can be generated. We first develop a protocol for preparing squeezed states using a sequence of displaced parity measurements applied along the anti-squeezed quadrature (see Fig.~\ref{fig1}(a)). The key observation underlying this approach is that a squeezed vacuum state occupies only the even-parity subspace of the Fock basis. In the limit of infinite squeezing, the state becomes invariant under displacements along the squeezed axis, which leads to a set of eigenvalue conditions that can be implemented through parity measurements.
We analyze the performance of this scheme by quantifying the achievable squeezing and assessing the impact of realistic imperfections, demonstrating the robustness of the protocol.
We further show that sequences of displacements combined with parity measurements can be interpreted as operations that generate coherent superpositions of a state and its phase-space–reflected counterpart. This observation enables a natural extension of the protocol to convert a squeezed state into a Gottesman–Kitaev–Preskill (GKP) state, which may be viewed as a comb of displaced squeezed states.
Finally, we show that the combination of dispersive measurements and displacements is universal in that it can generate an arbitrary state.

We note that existing methods for measurement-based preparation of squeezed states typically require measuring phase-space quadratures of the bosonic state \cite{vanner2011pulsed, szorkovszky2011mechanical, wollman2015quantum}.  In several qubit-mode systems these are observables that are not native elementary measurement primitives, and may be difficult to implement, motivating our work.  We point out however that recently there have been optimized-control techniques that have shown progress demonstrating that such observables can be accessed through ancillary-qubit measurements~\cite{krisnanda2026direct}.

\begin{figure*}[t]
\includegraphics[width=\linewidth]{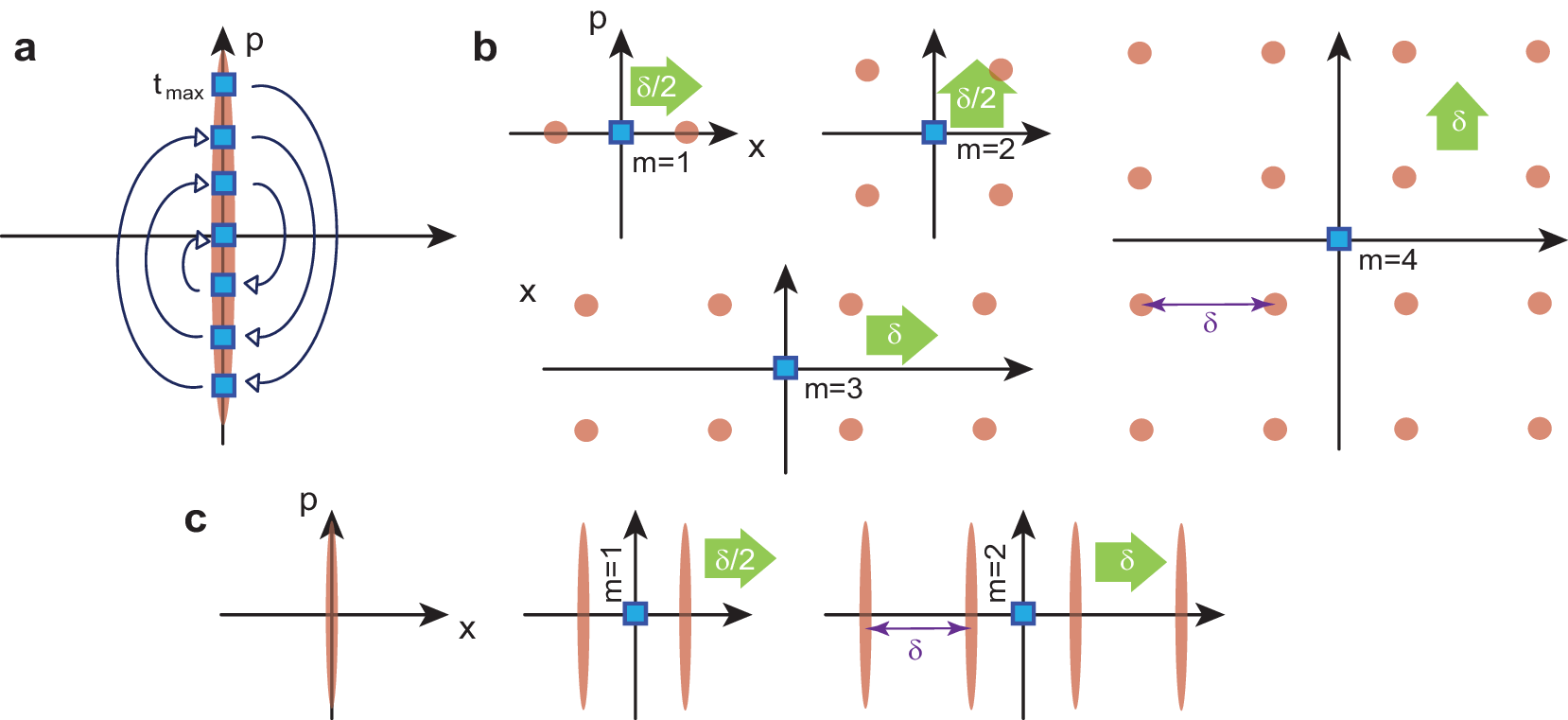}
\caption{Protocols for generating (a) squeezed states, (b)  multi-component cat states, and (c) GKP states with parity measurements and displacements.  Squares denote the displaced parity measurement $ P_+ $ as given in (\ref{dispparity}), where the location of the parity measurement is at $ \alpha $.  Circles and ovals represent coherent and squeezed states respectively. Large arrows denote displacements (\ref{dispop}), with the magnitude of the displacement indicated within the arrow, and the direction of the displacement is in the arrow direction.  (a) The squeezed state measurement sequence occurs in the order indicated by the curved arrows.  (b) In each measurement step, the displacement is made before the parity measurement, creating a square lattice with spacing $ \delta $.   (c) In first step, create a squeezed state using the sequence given in (a).  Then use a sequence of displacements and parity measurements as shown to point reflect the squeezed state to form a comb separated by $ \delta $.    
}
\label{fig1}
\end{figure*}

\section{Physical system}

\subsection{Definitions}

The physical system that we consider in this paper is a single bosonic mode, which is associated to the annihilation operator $ a $ satisfying the canonical commutation relation $ [a, a^\dagger ] = 1 $. 
This mode can represent a continuous-variable degree of freedom in different physical systems, such as microwave photons in cQED, or phonons in cQAD or trapped ions. We keep our formalism general such that it can be applied to any such platform.  

The Hilbert space of a bosonic mode is spanned by Fock states, namely the states defined as
\begin{align}
|n \rangle = \frac{{(a^\dagger)}^n}{\sqrt{n!}} | 0 \rangle ,
\label{fockstates}
\end{align}
where $ | 0 \rangle $ is the vacuum state satisfying $a\ket{0}=0$. Coherent states are defined as 
\begin{align}
| \alpha \rangle = e^{-|\alpha|^2/2} e^{\alpha a^\dagger} | 0 \rangle
= e^{-|\alpha|^2/2} \sum_{n=0}^\infty \frac{\alpha^n}{\sqrt{n!}} |n \rangle .
\label{coherent}
\end{align}
Since bosonic degrees of freedom arise from harmonic oscillators, it is natural to introduce position and momentum quadrature observables as
\begin{equation}
    x = \frac{1}{2} ( a+ a^\dagger) \;,\qquad
    p = -\frac{i}{2} ( a - a^\dagger)   \;,
\end{equation}
such that $ [x,p ] = i/2 $.  These have eigenstates \cite{braunstein2005quantum}
\begin{align}
    x|x_0 \rangle & = x_0 | x_0 \rangle  \nonumber \\
    p|p_0 \rangle & = p_0 | p_0 \rangle   ,
    \label{quadeigenstates}
\end{align}
satisfying the orthogonality conditions
\begin{align}
 \langle x_0 | x_0' \rangle &  = \delta (x_0 - x_0' ) \nonumber \\
 \langle p_0 | p_0' \rangle &  = \delta (p_0 - p_0' )  ,
\end{align}
and correspond to infinitely squeezed states with precisely defined position $ \langle x \rangle  = x_0 $, or momentum $ \langle p \rangle  = p_0 $, respectively.  We may also  consider a generalized quadrature parametrized by an angle $\theta$ in phase space as
\begin{align}
x^{(\theta)} = \frac{1}{2} (a e^{-i \theta} + a^\dagger e^{i \theta})
\end{align}
with eigenstates
\begin{align}
x^{(\theta)}| x^{(\theta)}_0 \rangle = x^{(\theta)}_0| x^{(\theta)}_0 \rangle .  
\label{thetaeigen}
\end{align}
Of particular relevance is the squeezed vacuum state, for which $\langle x^{(\theta)}\rangle=0$, which is defined as \cite{gerry2023introductory}

\begin{equation}
|\xi \rangle = \frac{1}{\sqrt{\cosh r}} 
\sum_{n=0}^\infty  (-1)^n \frac{\sqrt{(2n)!}}{2^n n!} e^{2in \phi} \tanh^n r |2 n \rangle \;,
\label{squeezedstate}
\end{equation}
where  $ \xi = r e^{2 i \phi} $.  This state can be understood as the action of the squeezing operator 
\begin{align}
    S(\xi) = e^{\frac{1}{2} ( \xi^*  a^2 - \xi {a^\dagger}^2)} \;,
\end{align}
on the vacuum state, i.e. $|\xi \rangle  = S(\xi) | 0 \rangle \nonumber$.

Defining the variance as
\begin{align}
    \text{Var} (x) = \langle x^2 \rangle -  \langle x \rangle^2 \;,
\end{align}
we have that the vacuum state $\ket{0}$ has isotropic variance $\text{Var}(x^{(\theta)}) = 1/4$, while the squeezed state $\ket{\xi}$ has 
$\text{Var}(x^{(\theta)}) = 1/4\left( \cosh (2r) - \sinh(2r) \cos2(\phi-\theta) \right) $. 
Along the direction defined by {\color{red} $\theta=\phi$} the variance reads $\min_\theta \text{Var}(x^{(\theta)})= e^{-2r}/4$, showing a reduction below the quantum noise of the vacuum state.
This reduction is often quantified in decibel (dB) from the formula
\begin{align}
S_{\text{dB}} = -10 \log_{10} \frac{ \min_{\theta} \text{Var} (x^{(\theta)}) }{\text{Var}_0 (x) }.
\label{dbsqueezing}
\end{align}
Here, the variance $\text{Var}(x^{(\theta)})$ is calculated with respect to the final state, while $\text{Var}_0 (x) = 1/4$ denotes the variance of the vacuum state, independent of $\theta$.
%Since our protocols typically focus on squeezing the $x=x^{(\theta=0)}$ quadrature, we can often omit the minimization over $\theta$.

\subsection{Control and measurement}

We assume that the bosonic mode under investigation can be controlled and measured in some elementary ways. 
The main operation we consider is a displacement in phase space, as described by the operator
\begin{align}
    D(\alpha) = e^{\alpha a^\dagger - \alpha^* a}  .
    \label{dispop}
\end{align}
This is typically realized by driving the bosonic mode on resonance, e.g with a coherent microwave signal or laser field.
%for cQED \cite{blais2004cavity} and cQAD \cite{gustafsson2014propagating}, or a laser field in trapped ions \cite{wineland1979laser}. 

The measurement we will consider are parity measurements $P=(-1)^{\hat{n}}=P_+-P_-$,  with $ \hat{n} = a^\dagger a $, where 
\begin{subequations}\label{paritymeas}
\begin{align}
P_{+} &  = \sum_{n \in \text{even}} | n \rangle \langle n |  \\
P_{-} &  = \sum_{n \in \text{odd}} | n \rangle \langle n | 
\end{align}
\end{subequations}
represent projectors into the subspaces with even/odd parity, respectively. 
Parity measurements are routinely performed in spin-boson systems, where the bosonic mode is coupled to a two-level (qubit) degree of freedom \cite{LutterbachPRL97,Vlastakis2013,blais2021circuit,von2022parity}. 
In such platforms, the interaction is described by the Jaynes-Cummings model
\begin{align}
H/\hbar  =  \omega_a a^\dagger a + \frac{ \omega_0}{2} \sigma^z + g ( a \sigma^+ + a^\dagger \sigma^- )  ,
\end{align}
where $ \omega_a $ is the frequency of the bosonic mode, $ \omega_0$ is the frequency of the qubit, and $ g $ is the qubit-boson coupling rate.  In the dispersive regime ($ \Delta \gg |g| $) and putting the boson and qubit in the rotating frame with frequency $ \omega_a $, the effective Hamiltonian reads
%
%\begin{align}
%H/\hbar \approx \omega_a a^\dagger a - \frac{|g|^2}{\Delta} a^\dagger a \sigma^z
%\end{align}
\begin{align}
H/\hbar \approx - \frac{|g|^2}{\Delta} \hat{n} \sigma^z  - \frac{\delta}{2} \sigma^z,
\end{align}
where $ \Delta = \omega_a -  \omega_0 $ is the qubit-boson detuning and $ \delta = \Delta + \frac{|g|^2}{\Delta}  $.  Initializing the qubit in the state $ | + \rangle = (|0 \rangle + |1 \rangle )/\sqrt{2} $, letting it interact with the mode according to the dispersive Hamiltonian for a time $t_\text{int}$, and projecting on the qubit in the $ \sigma^x $ basis
results in an effective measurement of the bosonic mode that reads \cite{feng2025quantum} 
%
%\begin{align}
%{\cal M}_{+} (\tau) &  = \sum_n e^{-in \omega t} \cos (n \tau) |n \rangle \langle n | \nonumber \\
%{\cal M}_{-} (\tau) &  = \sum_n e^{-in \omega t} \sin (n \tau) |n \rangle \langle n |
%\end{align}
\begin{align}
{\cal M}_{+}(\tau, \phi)   &  = \sum_n \cos (n \tau + \phi) |n \rangle \langle n | \nonumber \\
{\cal M}_{-} (\tau, \phi)   &  = \sum_n \sin (n \tau + \phi) |n \rangle \langle n | ,
\label{measurementops}
\end{align}
with $\tau = |g|^2 t_\text{int} /\Delta$ and $ \phi = \delta t_\text{int}$. 
These are the projectors corresponding to measuring the qubit in the $\ket{\pm }$ state, respectively.
 For the choice $ \tau = \pi/2, \phi = 0  $,  %and $ \omega t_\text{int} = \pi/2 $
the measurement operators reduce to the even and odd parity projectors given in (\ref{paritymeas}), up to a phase which may be removed by an additional unitary  
$ e^{i \pi a^\dagger a/2 } $.

\section{Squeezing with parity measurements}

\subsection{Basic idea}

Our first goal is to prepare the bosonic mode in a squeezed state by using a measurement-based state preparation protocol that relies only on displacement operations and parity measurements.
To understand the idea behind our approach, we first observe that that the squeezed vacuum (\ref{squeezedstate}) has an even parity, hence it is an eigenstate of the even parity operator
\begin{align}
P_+ |\xi \rangle = |\xi \rangle  .
\end{align}
Second, for an ideal infinitely squeezed state, it is invariant under translations along the anti-squeezed axis
\begin{align}
D( \alpha = -it) |x_0  \rangle & = e^{-i ( a+a^\dagger) t} |x_0  \rangle = e^{-i2 x_0 t} |x_0  \rangle 
%D( \alpha = -t) |p_0  \rangle & = e^{ (a- a^\dagger ) t} |p_0  \rangle = e^{i2 p_0 t} |p_0  \rangle  
\end{align}
where we used Eq.~(\ref{quadeigenstates}) and $ t $ is a real (unitless) parameter. 
Similar relations also hold for squeezing along other directions, but for simplicity we will consider here $x$-squeezed states, without loss of generality. 
Defining the displaced parity measurement as
\begin{align}
P_{\pm} \left( \alpha \right)=D (\alpha)  P_{\pm} D^{\dagger}  (\alpha) ,
    \label{dispparity}
\end{align}
it follows that for displacements along the anti-squeezed axis we have 
\begin{align}
P_+ ( \alpha = it) |x_0  \rangle & = |x_0  \rangle .
%P_+ ( \alpha = t) |p_0  \rangle & = |p_0  \rangle  .
\end{align}
Performing multiple projections along the anti-squeezed axis thus gives
\begin{align}
\prod_{m=1}^{M} P_+ ( \alpha = it_m) |x_0  \rangle = |x_0  \rangle .
\label{measseq}
\end{align}
Here, $ M $ is the number of measurements that are performed. We note that it is also possible to write this in a Hamiltonian formulation (see Appendix \ref{sec:hamsqueezed}).  

When applied to an arbitrary initial state {\color{red} $\ket{\psi_0}$,} the parity measurement (\ref{dispparity}) results in a projection onto the even/odd parity subspace. 
For non-commuting displaced-parity measurements (\ref{dispparity}), the sequence of $M$ measurements in (\ref{measseq}) results in a series projections onto an increasingly narrow subspace. For $M$ sufficiently large, we expect that any initial state {\color{red} $\ket{\psi_0}$,} converges to a squeezed state.

\subsection{Squeezing protocol with postselection}

After having introduced the basic idea behind our squeezing protocol, we now formulate it more rigorously.
Let us define the projector for a sequence of even parity projections along the desired squeezing orientation $\theta $ as
\begin{align}
{\cal P}_\theta(\vec{t})  = \prod_{m=1}^{M} P_+ ( \alpha = i e^{i\theta} t_m )
\label{bigprojdef}
\end{align}
where $ \vec{t} = (t_1, t_2, \dots, t_{M}) $ is a sequence of adimensional time parameters. 
The order of operators in the product proceeds from right to left, such that the operator with $t_1$ is applied first and $t_{M} $ last to the state.  
We note that, since the measurement sequence (\ref{bigprojdef}) involves only the even parity outcomes, to implement this operator in practice postselection is required.    

An equivalent expression for the operator (\ref{bigprojdef}) is
\begin{align}
{\cal P}_\theta(\vec{t}) & =  D (i e^{i\theta} t_{M}) \left[ \prod_{m=1}^{M} P_+ D(-i e^{i\theta} \Delta t_m)  \right] 
\label{bigprojdef2}
\end{align}
where the difference between displacements is
\begin{align}
\Delta t_m = t_m - t_{m-1} ,
\end{align}
and we take $ t_0 = 0 $. 
In the above, we have used the identities for the displacement operators $ D^\dagger (\alpha) = D( - \alpha) $ and the formula for combining displacements \cite{gerry2023introductory}
\begin{align}
D(\alpha) D (\beta) = e^{i \text{Im} (\alpha \beta^*)} D(\alpha + \beta) .
\label{dispdisp}
\end{align}
Eq. (\ref{bigprojdef2})  is simpler from an implementation standpoint, as it simplifies two consecutive displacements. 

We then claim that for a suitable choice of displacements $ \vec{t} $
\begin{align}
| {\cal P}_\theta(\vec{t}) \rangle := 
\frac{{\cal P}_\theta(\vec{t}) |0 \rangle}{\sqrt{ p_{\text{suc}} (\theta, \vec{t}) } } \approx | x^{(\theta)}_0 = 0  \rangle  ,
\label{claim}
\end{align}
where the state that is approached is the squeezed state (\ref{thetaeigen}).  
Here we defined the success probability
\begin{align}
p_{\text{suc}} (\vec{t}) = \langle 0 | {\cal P}^\dagger _\theta(\vec{t}) {\cal P}_\theta(\vec{t}) |0 \rangle .
\label{totalsuccessprob}
\end{align}
The success probability of the $m$th measurement can be written as 
\begin{align}
p^{\text{suc}}_m = \frac{p^{\text{suc}} (\vec{t}_m)}{p^{\text{suc}} (\vec{t}_{m-1})}
\label{sucprobn}
\end{align}
where $ \vec{t}_m = (t_1, t_2, \dots, t_{m}) $ is $ \vec{t} $ truncated to the $ m$th measurement.  

Due to the probabilistic nature of the parity measurement outcomes, the performance of our protocol is highly dependent upon the choice of $ \vec{t} $. In the next section, we show some concrete choices of $ \vec{t} $ which have a good performance in terms of generating squeezed states.

\section{Numerical evaluation of parity measurement based squeezing}

\subsection{Method}

We numerically simulate our protocol to illustrate its performance for generating of squeezed states.
To include loss, which is the main decoherence channel in many experimental realizations, we express Eq.~(\ref{claim}) in the density matrix formalism. For this, we note that each measurement step corresponds to the operation
\begin{align}
\rho^{(m)} = \frac{ \sum_k E_k P_+ ( \alpha = i e^{i \theta} t_m ) \rho^{(m-1)} P_+^\dagger ( \alpha = i e^{i \theta} t_m ) E_k^\dagger }{\text{Tr} ( P_+ ( \alpha = i e^{i \theta} t_m ) \rho^{(m-1)} P_+^\dagger ( \alpha = i e^{i \theta} t_m )} ,
\label{measurementcycle}
\end{align}
where $ \rho^{(m)} $ is the state after $ m $ measurements.
Here we introduced the Kraus operator for bosonic loss
\begin{align}
E_{k} = \sqrt{\frac{\left( 1 - \eta\right)^{k}}{k!}}\sqrt{\eta}^{a^{\dagger}a}a^{k}
, \label{14}
\end{align}
where $k$ is the number of lost bosons, and $\epsilon = 1- \eta$ denotes the probability of losing a photon within each measurement step.

%When including bosonic loss we use the density matrix formalism with the state initialized in $  | 0 \rangle \langle 0 | $.  Then, for each measurement we apply the operation

We note that the bosonic loss operator is applied after {\it each} measurement in Eq.~(\ref{measurementcycle}).
Thus after $ M $ measurements, there is a probability of $ 1- \eta^M $ that at least one boson has been lost.
Thus for any $ \epsilon > 0 $, the probability that no error has occurred is exponentially diminishing.
We choose this type of model as it is closest to experimental reality, where each measurement requires finite time to be performed.

To characterize the state resulting from our protocol, we plot its Wigner function.
This is also what is typically done experimentally, using the fact that the Wigner function can be expressed as the expectation value of the displaced parity operator \cite{royer1977wigner}
\begin{align}
   W ( \alpha) = \frac{2}{\pi} \text{Tr} \left[ \rho (P_+(\alpha) - P_-(\alpha)) \right] ,
    \label{wignerfunc}
\end{align}
with the displaced parity operators defined in (\ref{dispparity}).  

Our simulations are expressed in the Fock basis using a cutoff at $\ket{n_{\text{cut}}}$, where $ n_{\text{cut}} $ is chosen large enough to ensure convergence.

\subsection{Measurement sequence}

We start discussing a suitable choice for the measurement parameters $ \vec{t} $, which determine the measurement sequence to be performed.  A simple yet effective ansatz for generating squeezing is
\begin{align}
    t_m = (-1)^{m-1} t_{\max} \left(1-\frac{  \lfloor (m-1)/2 \rfloor}{ \lfloor (M-1)/2 \rfloor} \right)  .
    \label{symmetric}
\end{align}
Here, $ \lfloor \cdot  \rfloor $ is the floor function and we only consider odd $ M $.  
The measurement pattern is illustrated in Fig. \ref{fig1}.  
The first measurement occurs at $ \alpha = t_{\max} e^{i \theta} $, followed by a measurement at $ \alpha = - t_{\max} e^{i \theta} $.  
The next two measurements are symmetrically distributed towards the origin but with smaller amplitudes, until the final measurement is at the origin.   

From an intuitive point of view, the measurement sequence (\ref{symmetric}) is a reasonable choice for the following reasons. 
First, it ensures that the final measurement occurs at the vacuum $ t_M = 0 $, which enforces the prepared state to have an even parity. 
This is consistent with any squeezed state (\ref{squeezedstate}) and it thus maximize the fidelity with respect to this class of states. 
Second, the measurements are symmetrically distributed about the origin along the anti-squeezing axis, reflecting the fact that squeezed vacuum states (\ref{squeezedstate}) are invariant under a rotation $ e^{i \pi a^\dagger a } $.  

Other choices of measurement patterns distributed along a line also produce squeezing, but (\ref{symmetric}) has a good balance between performance, simplicity, and success probability.
As an example, we show in Fig. \ref{fig2}(b) the Wigner function for the same number of measurements and amplitude, comparing (\ref{symmetric}) and a linear sequence 
\begin{align}
 t_m  = t_{\max} \left(1-\frac{m-1}{M-1} \right)  .
\label{linear}
\end{align}
We see that the ansatz (\ref{symmetric}) has a better performance in terms of generating a squeezed state. In Supplementary Information Sec. I we show the performance of alternative ansatz choices which can show superior squeezing, at the cost of a more complex sequence and a lower success probability. 
Henceforth,  we discuss the performance within the choice (\ref{symmetric}).

\begin{figure}[t]
\includegraphics[width=\linewidth]{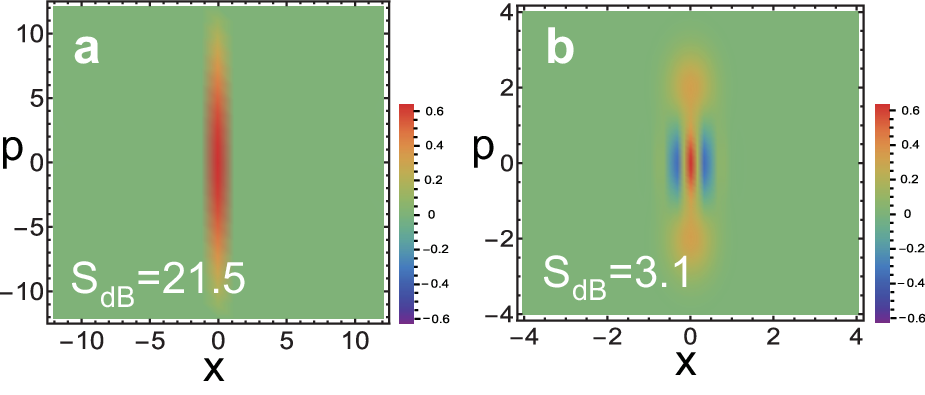}
\caption{Wigner functions of the squeezed state generated according to (\ref{claim}). We compare two measurement sequences: (a) symmetric ordering (\ref{symmetric}) and  (b) linear ordering (\ref{linear}). Wigner functions were calculated using (\ref{wignerfunc}).    We use $ t_{\max} = 2 $ and $ M = 11 $ for both cases.  
Cutoffs of (a) $ n_{\text{cut}} = 601 $ and (b) $ n_{\text{cut}} = 101 $ were used. 
\label{fig2}  }
\end{figure}

\subsection{Convergence to an infinitely squeezed state}
 
\begin{figure}[t]
\includegraphics[width=\linewidth]{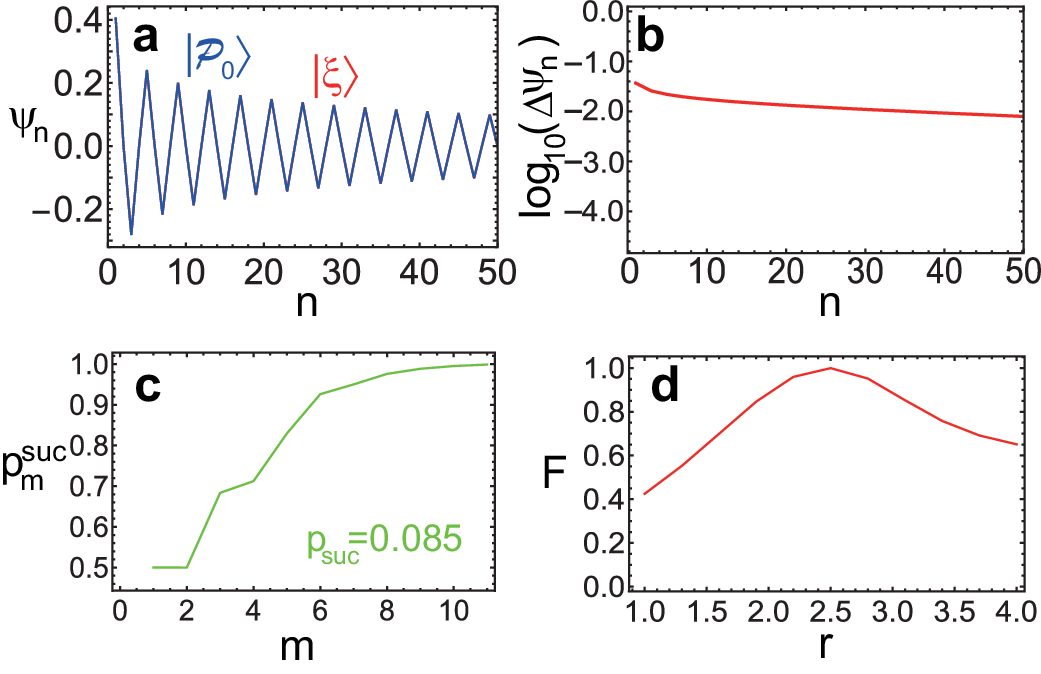}
\caption{Convergence of the state (\ref{claim}) towards the infinitely squeezed state.  We use the measurement sequence (\ref{symmetric}) with $ M = 11 $, $ t_{\max} = 2 $, $ \theta = 0 $.  
(a) Wavefunction $ \psi_n = \langle n | {\cal P}_0 (\vec{t}) \rangle $.  A comparison to the squeezed state with $ \xi = 2.5 $ is shown.  (b) 
The difference between the wavefunctions in (a), quantified by  $ \Delta\psi_n= |\langle n|{\cal P}_0(\vec t)\rangle -\langle n |  \xi=2.5 \rangle| $.}  (c) Success probability $ p_n^{\text{suc}} $. (d) Fidelity with squeezed states $ F = | \langle \xi = r | {\cal P}_0 (\vec{t}) \rangle |^2 $. We use a Fock state truncation of $ n_{\text{cut}} =601 $.  
\label{fig3}   
\end{figure}

We first demonstrate our claim (\ref{claim}), where for a suitable measurement pattern the state approaches an infinitely squeezed state.  
Figures \ref{fig2}(a) and \ref{fig3} shows some representative results with $ M = 11 $ measurements.  
From Fig. \ref{fig2}(a) we see an excellent amount of squeezing, reaching $ S_{\text{dB}} = 21.5 $ dB.  
In  Figure \ref{fig3}(a)(b), we show the Fock state amplitudes and compare it to a squeezed state with $ \xi $ chosen to most closely approximates our final state.  
We see a nearly indistinguishable amplitude distribution, corresponding to a high fidelity.
Fig. \ref{fig3}(c) shows the success probabilities for each measurement Eq.~(\ref{sucprobn}).  
Note how the first few measurements have typically $ p_m^{\text{suc}} \approx 0.5 $, but generally these converge towards $ p_m^{\text{suc}} \rightarrow 1 $ as $m$ increases.  
For this reason, the total success probability does not follow a simple exponential decay with $ M $; but it rather depends on the specific details of the sequence, including initial state and measurement points. 
For this example we find $ p_{\text{suc}} \approx 8.5\% $, which is significantly better than $1/2^M \approx 0.05 \%$.
Fig. \ref{fig3}(d) shows the fidelity with respect to the squeezed state (\ref{squeezedstate}) for a range of squeezing parameters.  
We see that the final state is very closely approximated ($ F = 0.998 $) by a squeezed state with $ \xi = 2.5 $.

Generally, we find that the larger the number of measurements $M $, the better the convergence towards the infinitely squeezed state is. 
The numerical limitation is however that an infinitely squeezed state has a non-negligible population of states with very high Fock number, eventually exceeding our Hilbert space truncation at $\ket{n_\text{cut}}$.
To see this, set $ r\rightarrow \infty $ in (\ref{squeezedstate}) and we have 
\begin{align}
|x_0 = 0 \rangle & \propto \sum_{n=0}^\infty  (-1)^n \frac{\sqrt{(2n)!}}{2^n n!} |2 n \rangle \nonumber \\
& \approx \sum_{n=0}^\infty\frac{(-1)^n }{(\pi n)^{1/4}} |2 n \rangle ,
\end{align}
which is slowly converging in Fock number.  With a larger number of measurements $ M $, the truncation effects start to play a role, effectively limiting the level of squeezing that can be obtained.  
For example, for $ M = 21 $ measurements with $ t_{\max} =  5 $ and a truncation $n_{\text{cut}} = 1001 $, we obtain a squeezing of up to $ S_{\text{dB}} = 26\,$dB.  
It is however likely the actual squeezing level is even larger. 

%\mf{Oh, but I would have expected a different statement: for example, if I truncate to 1000 I can consider only up to $M=$XX measurement in order to get a squeezing that is real and not limited by truncation. Can you make some plot for this??}  

\subsection{Small number of measurements}

Actual experimental implementations of our protocol are limited in the total number of measurements $M$, due to either constraints on time or tradeoff with losses. 
Therefore, it is of interest to know what level of squeezing can be attained for a small, finite, number of measurements.  
As mentioned above, in the limit of large number of measurements $ M $ the prepared state convergences towards a highly squeezed state.  
For a few number of measurements, however, we have to deal with a dynamical regime of the protocol, before convergence.  
For this reason, there is a strong dependence on the choice of initial state and $ t_{\max} $. 
In the following, we investigate this dependence, always considering the ideal scenario without losses. 
The effect of losses will then be discussed in detail in the next section.

Figure \ref{fig4} shows the effect of different choices of $ t_{\max} $ for a fixed small number of measurements $ M = 3 $.  
In Figures \ref{fig4}(a)-(c), we show in the Fock basis the final states obtained from three different $ t_{\max} $, together with the closely matching squeezed state.  
Importantly, we see that there is an optimal value of $ t_{\max} $ which maximizes the amount of squeezing.  
In fact, choices of $ t_{\max} $ that are too small result in almost no squeezing, since all projections occur near the origin in phase space.  
This is clear for $ t_{\max} =0  $, where all the displaced parity measurements are simply $ P_+ $ and the final result is still the vacuum state.  
For a choice of $ t_{\max} $ that is too large, the obtained state shows a large disagreement with the family of squeezed states, mainly due to the lack of convergence for $M$ small.  
At the optimal $ t_{\max} $, the state closely approximates a squeezed state with a large fidelity of $F \approx 0.99 $.  
In Fig. \ref{fig4}(d) we evaluate the squeezing with (\ref{dbsqueezing}) and we obtain a maximum value of $ S_{\text{dB}} \approx 8.9 $. 
The Wigner function for the optimized case is shown in Fig. \ref{fig4}(e), where we see the expected squeezed distribution.  
Some minor Wigner negativity is seen due to the finite number of measurements.  
The total success probability is shown in Fig. \ref{fig4}(g).   
For the optimum $ t_{\max} $, $ p_{\text{suc}} = 0.32 $.
Thus, as long as the maximum displacement $ t_{\max} $ is optimized, even a very small number of measurements can result in the preparation of a state with a significant amount of squeezing with large success probability.

\begin{figure}[t]
\includegraphics[width=\linewidth]{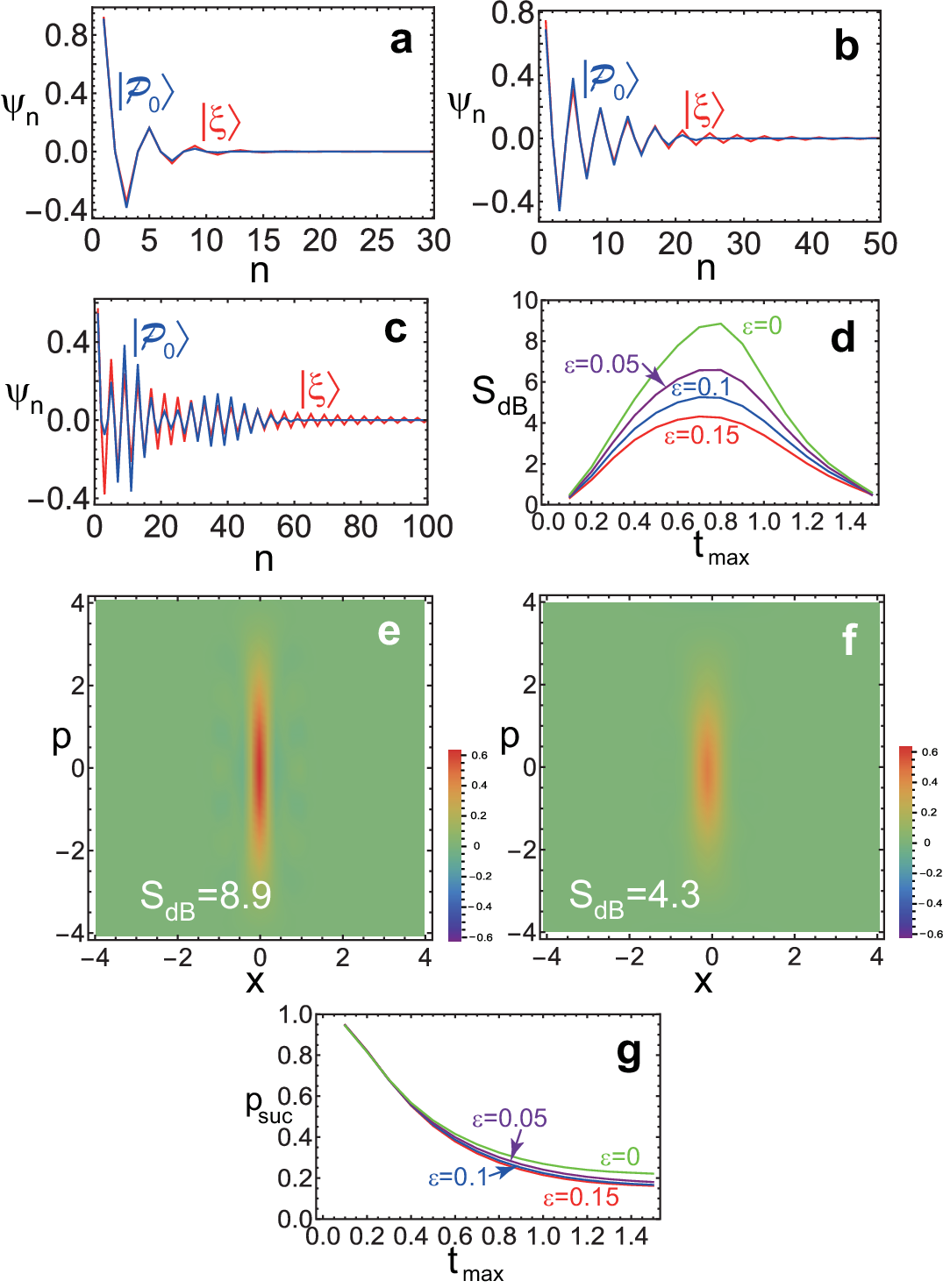}
\caption{Effect of the amplitude of displacement $ t_{\max} $ with a small number of measurements. 
We evaluate (\ref{claim}) with the measurement sequence (\ref{symmetric}) with $ M = 3 $ and $ \theta = 0 $. 
Wavefunction $ \psi_n = \langle n | {\cal P}_0 (\vec{t}) \rangle $ for (a) $ t_{\max} = 0.4 $; (b) $ t_{\max} = 0.8 $; (c) $ t_{\max} = 1.5 $.  A comparison to the squeezed state with  (a) $ \xi = 0.6 $; (b) $ \xi = 1.2 $; (c) $ \xi = 1.8 $  is shown. (d) The squeezing (\ref{dbsqueezing}) as a function of $ t_{\max} $.
The Wigner distribution using the optimal $ t_{\max} = 0.8 $ for (e) $ \epsilon = 0 $ and (f) $ \epsilon = 0.15 $.  (g)  The total success probability (\ref{totalsuccessprob}) as a function of $ t_{\max} $.  The photon loss probability $ \epsilon = 1 - \eta $ are shown for (d)(e). For (a)(b)(c)(e) the photon loss probability is $ \epsilon = 0 $.   We use a Fock state truncation of $ n_{\text{cut}} =201 $ for the pure state calculations and $ n_{\text{cut}} =51 $ for the mixed state calculations.  \label{fig4}   }  
\end{figure}

% For the optimal displacement amplitude \(t_{\max}=0.8\), the success probability is \(p_{\rm suc}\simeq 0.32\). 

% The photon loss probabilities \(\epsilon=1-\eta\) are indicated in (d) and (g). 

\subsection{Effect of losses}
\label{sec:loss}

In many experimental platforms of interest for our work (e.g. photonic, phononic), the dominant decoherence mechanism is energy relaxation, meaning the loss of excitations towards the environment.
Hence, it is of crucial importance to investigate whether squeezing may be realized under realistic circumstances that include inevitable imperfections.  
To this end, we repeat the analysis presented in the previous section including losses in the dynamics, using the density matrix formulation (\ref{measurementcycle}).  

In Fig. \ref{fig4}(f) we show the effect of losses on the Wigner function, always for $ M = 3 $ measurements. 
Compared to the lossless case, Fig. \ref{fig4}(e), it is visible that the squeezing is diminished, with a broader distribution having a lower amplitude at the origin. 
The squeezing reduction is also visible in Fig. \ref{fig4}(d), where for 15 \% loss per measurement the squeezing reduces to 4.3\,dB.   
Interestingly, the optimal value of displacement $ t_{\max} $ appears to be largely unaffected by the presence of loss.  
The total success probability shown in Fig. \ref{fig4}(g) is also largely unaffected by the presence of losses.  

Figure \ref{fig5} shows the effect of losses for the case of a larger number of measurements, $ M = 11 $.  For the Fock state populations shown in Fig. \ref{fig5}(a), one major difference with the pure state case (see Fig. \ref{fig3}(a)) is that high photon numbers are suppressed.  
This occurs due to the bosonic amplification of losses for larger Fock number states \cite{scully2009super,byrnes2011accelerated}, which come as an advantage in numerical simulations since it allows us to set the Hilbert space truncation to a lower $ n_{\text{cut}} $ value without affecting the accuracy of the results.  
Evaluating the squeezing, we note that the presence of losses causes a significant drop in the squeezing level, see Fig. \ref{fig5}(b). 
This is also due to the fact that for a larger number of measurements $M$ losses are more pronounced, since $ \epsilon $ is the loss probability {\it per} measurement.  
For example, for $ \epsilon = 0.01 $, this corresponds to a total loss probability of $ 1 - \eta^M \approx 0.1 $ for $ M = 11$. 
The introduction of loss creates an optimal value of $ t_{\max} $ around $ t_{\max}\approx 1 $, whereas for $ \epsilon = 0 $ the optimal value seems to occur outside our simulation parameter regime, for $ t_{\max}> 2 $. 
The total success probabilities, which we show in Fig. \ref{fig5}(c), are also more affected by losses compared to the case with $M=3$ shown in \ref{fig4}(g). 
This is because here the loss probability is higher, and thus also the probability to accidentally flip the parity of the state during the protocol.  
Finally, we show in Fig. \ref{fig5}(d) the Wigner function of the resulting state. 
While not attaining the impressive squeezing of the lossless case, see Fig. \ref{fig5}(b), considerable squeezing may still be obtained in the presence of losses.

\begin{figure}[t]
\includegraphics[width=\linewidth]{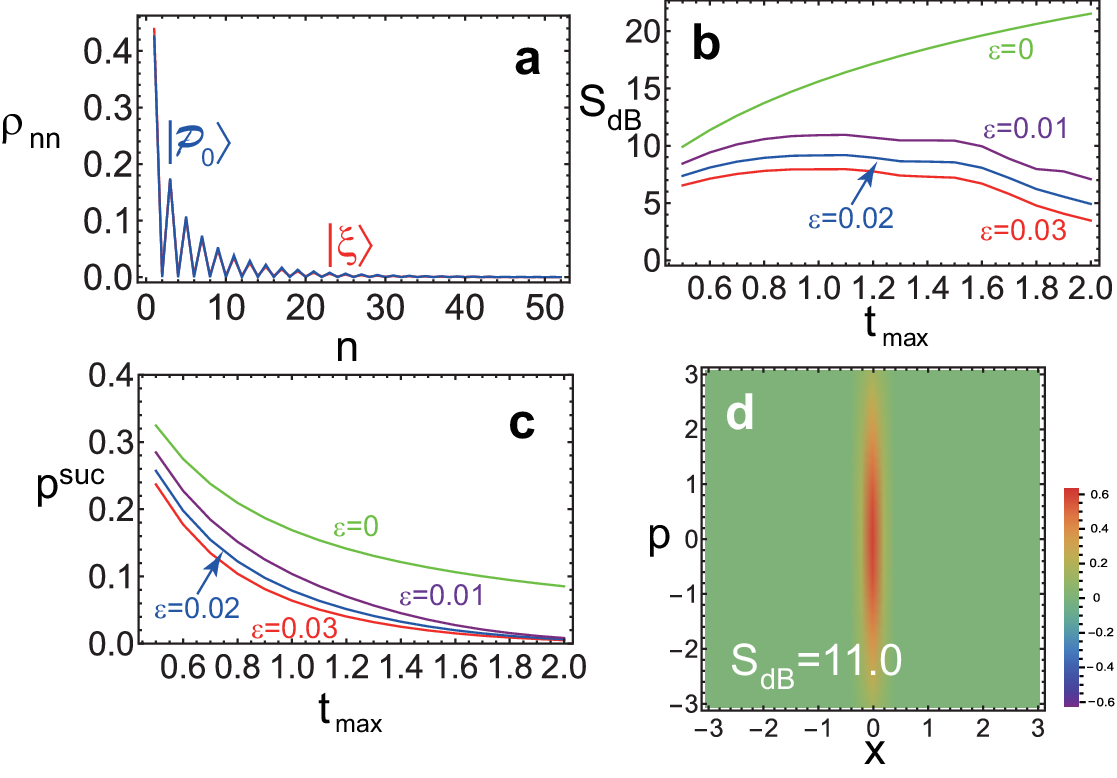}
\caption{Effect of bosonic loss on a larger number of measurements $ M = 11 $.  (a) The diagonal density matrix elements $ \rho_{nn} = \langle n | \rho |n \rangle $ versus the Fock number for $ t_{\max} = 1.1$, $ \epsilon = 0.01 $.  The squeezed state that is the closest match to the obtained state with $ \xi = 1.5 $ is shown for comparison. (b) The attained squeezing for various $ t_{\max} $ and the loss probabilities $ \epsilon $ as shown.  (c) The total success probability as a function of $ t_{\max} $ for the loss probabilities $ \epsilon $ as marked.  (d) The Wigner function of the final state for $ t_{\max} = 1.1$, $ \epsilon = 0.01 $.  For pure state ($\epsilon = 0 $) calculations $ n_{\text{cut}} = 601 $ is used, for mixed state ($\epsilon > 0 $) calculations  $ n_{\text{cut}} = 51 $ is used.  
\label{fig5}   }  
\end{figure}

\subsection{Scaling with the number of measurements}
\label{scaling}

In the lossless case, greater squeezing can be attained by increasing the number of measurements $ M $, since the projection sequence (\ref{measseq}) increasingly projects the initial state towards the infinitely squeezed state.  When including loss, there is a trade-off, as each measurement introduces more decoherence.  Here we examine how our protocol scales as the number of measurements $M $ is increased.  

Figure \ref{A1} shows the scaling of our protocol (\ref{measurementcycle}) with the number of measurements $M$. 
For each $(M,\epsilon)$ pair, we optimize $t_{\max}$ with respect to the squeezing of the final state. 
In the absence of loss, the optimal $t_{\max}$ grows with $ M $, 
corresponding to the larger variance of the state in the anti-squeezed direction.  As expected, this results in higher squeezing for larger $ M $.  
When including photon loss, the optimal $t_{\max}$ remains at much lower levels, as also seen in Fig. \ref{fig5}(b).  The squeezing remains high for small $\epsilon $, but decreases as more loss is introduced, especially for longer measurement sequences. Interestingly, the decrease in squeezing as $ M $ is increased is not very severe, with only a gradual decline.  Thus as long as there is a sufficient number of measurements to produce a high squeezing level for a given $ \epsilon $ (for the parameters in Fig.\ref{A1}(b), $ M \sim 5 $), there is only a weak dependence for a larger number of measurements.

\begin{figure}[t]
\includegraphics[width=\linewidth]{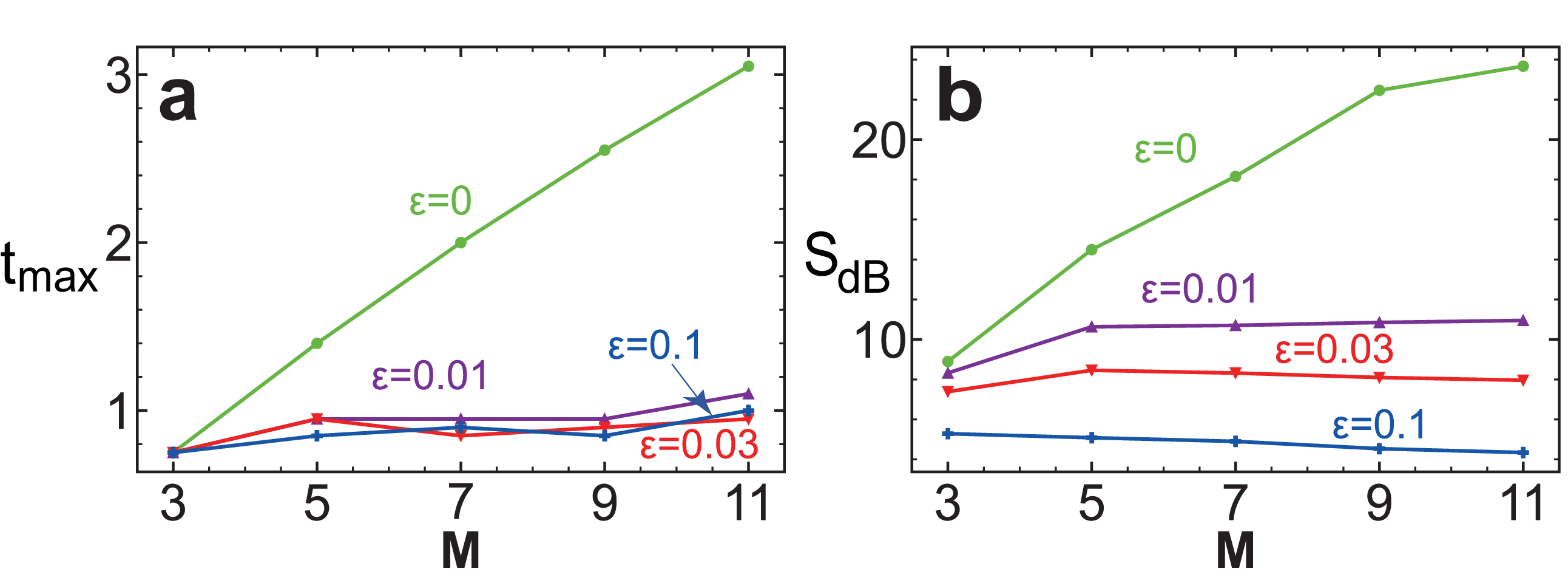}
\caption{Scaling of the optimized displacement $ t_{\max} $ and squeezing. (a) Optimized displacement $t_{\max}$ as a function of the number of parity measurements $M$ for different photon-loss probabilities $\epsilon$. For each pair $(M,\epsilon)$, $t_{\max}$ is optimized by maximizing the final squeezing.
    (b) Squeezing of the final state evaluated at the optimized $t_{\max}$.  
\label{A1}   }  
\end{figure}

\section{Point reflection based state generation}

 In this section, we show that the parity measurements more generally create point reflections in phase space.  We illustrate this method by creating  multi-component Schrodinger cat states and GKP states, which play a crucial role for error-protected quantum information encoding and quantum metrology.

\subsection{Multi-component cat states}

The key observation we take inspiration from is the fact that even/odd cat states $\ket{\alpha}\pm\ket{-\alpha}$ can be prepared by performing parity measurements on the coherent state $\ket{\alpha}$ {\color{red} \cite{brune1992manipulation}.}
This is because a coherent state can be written as the sum of even/odd cat states with a well defined parity, $\ket{\alpha} \simeq (\ket{\alpha}+\ket{-\alpha})+(\ket{\alpha}-\ket{-\alpha})$, meaning that a parity measurement results in a projection onto one of these two cat states.
More formally, we can write
\begin{align}
P_{\pm} | \alpha \rangle = \frac{1}{2} ( | \alpha \rangle \pm | - \alpha \rangle ) , 
\label{catparity}
\end{align}
which can be verified using (\ref{coherent}). Measurement-induced cat state generation was experimentally demonstrated in in numerous works \cite{deleglise2008reconstruction,vlastakis2013deterministically,wang2016schrodinger,pan2023protecting}.  

From a geometrical perspective, we can imagine that a parity measurement at the origin results in a point reflection around the origin, such that for a coherent state $\ket{\alpha}$ a new coherent state $\ket{-\alpha}$ is spawned.
Therefore, by iterating displacements and parity measurements according to 
\begin{align}
| {\cal C} (\vec{\alpha})  \rangle = \left[  \prod_{m=1}^M P_+ D(\alpha_m) \right]  | 0 \rangle
\label{multicat}
\end{align}
various superpositions of coherent states, i.e. multicomponent cat states, can be generated. 
Here, $ \vec{\alpha} = (\alpha_1, \alpha_2, \dots, \alpha_M) $ specifies the displacement sequence that is performed. 
Performing $ M $ measurements produces a superposition of up to $ 2^M $ coherent state, due to each measurement doubling the number of coherent states according to Eq.~(\ref{catparity}). 
Note that, this is different from the squeezing sequence Eq.~(\ref{bigprojdef}), as the displaced parity measurements involve two displacements, c.f. Eq.~(\ref{dispparity}).  

The result of Eq.~(\ref{multicat}) can be evaluated by combining Eq.~(\ref{catparity}) and 
Eq.~(\ref{dispdisp}).  
The latter equation involves a phase, coming from the Baker–Campbell–Hausdorff (BCH) formula, hence Eq.~(\ref{multicat}) will produce a superposition of coherent states which, in general, will involve nontrivial phase factors.  
For example, for $ M =2 $ measurements we have
\begin{align}
| {\cal C} (\vec{\alpha}) \rangle = & \frac{1}{4} \Big( e^{i \phi} | \alpha_1 + \alpha_2 \rangle + e^{-i \phi}   | \alpha_1  - \alpha_2 \rangle  \nonumber \\
& +e^{-i \phi}  | - \alpha_1 +  \alpha_2 \rangle + e^{i \phi}  | - \alpha_1  - \alpha_2 \rangle  \Big)
\label{catseq}
\end{align}
where $ \phi = \text{Im} ( \alpha_2 \alpha_1^*)  $.

By choosing the displacements $ \vec{\alpha} $ appropriately, some interesting states can be generated.  
In Figure \ref{fig1}(b), we show a scheme for making a square lattice of coherent states.  
The measurement sequence corresponds in this case to 
\begin{align}
\alpha_m = i^{\tfrac{1 + (-1)^m}{2} }  2^{\lfloor (m-1)/2 \rfloor -1 } \delta .
\end{align}
Each measurement doubles the lattice size, thus resulting very quickly in the creation of large lattices.  
Figure \ref{fig6}(a) shows the Wigner function for $ M = 2 $ measurements.  
Due to the BCH phases, in order to obtain a superposition with an even phase, it is convenient to choose a lattice spacing of $ \delta = 2 \sqrt{\pi} $.  
Since the BCH phases can be also interpreted as a Berry phase, if the area swept out by multiple displacements is a multiple of $ 2 \pi $ then the phases are removed. 
Further phase control can be performed by using selective number dependent arbitrary phase (SNAP) methods \cite{heeres2017implementing}.  
For example, if the phase of one coherent needs to be adjusted, then it can be displaced to the vacuum, then a SNAP operation can selectively adjust the phase of the $ | 0 \rangle $ state.

\begin{figure}[t]
\includegraphics[width=\linewidth]{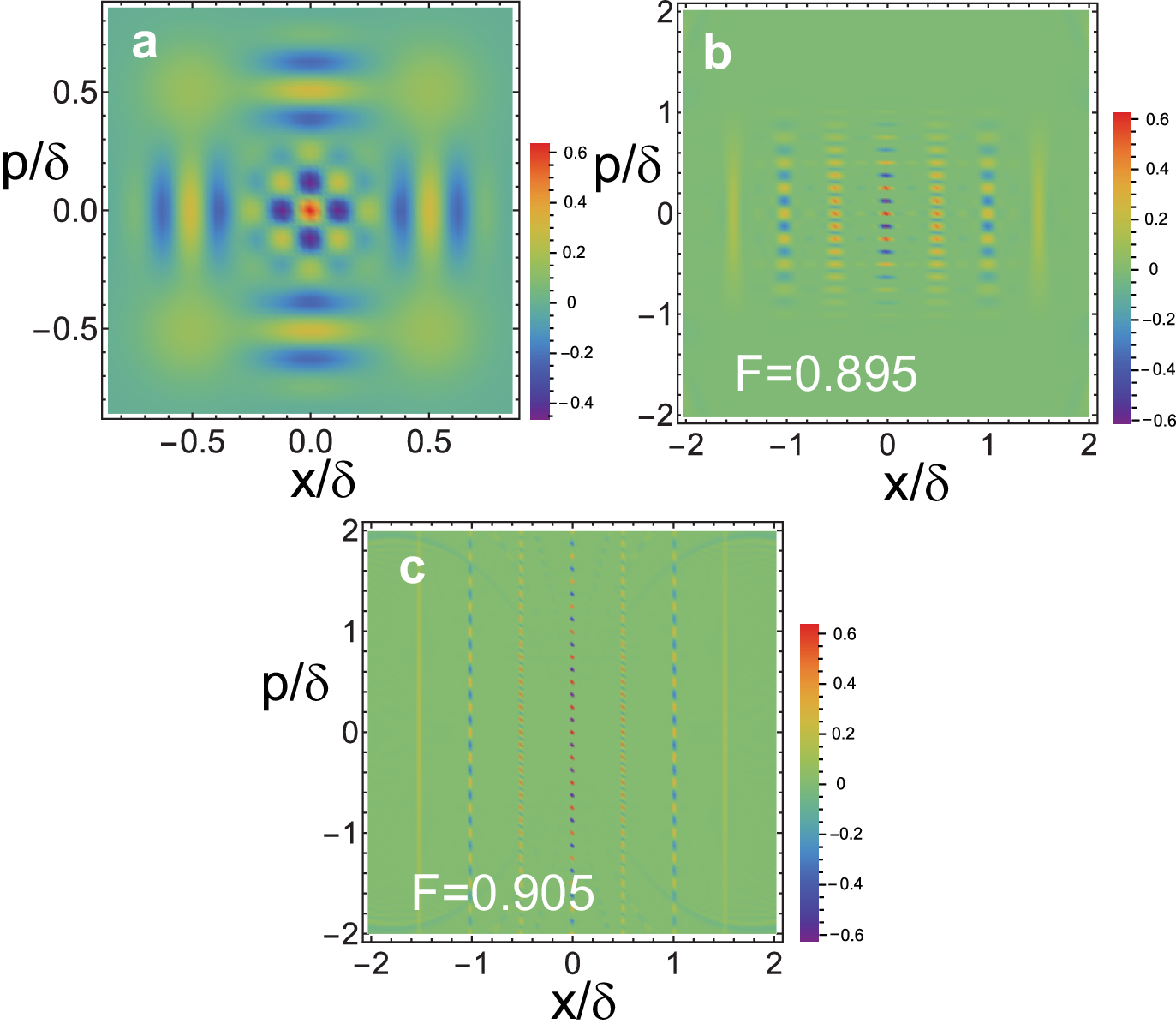}
\caption{Wigner functions for (a) multi-component cat states (\ref{multicat}) and (b)(c) GKP codeword $ |1_L \rangle $ given in  (\ref{gkpstate}).  To generate the state in (a), the sequence (\ref{multicat}) is used with $ M =2 $ and $ \vec{\alpha} = \{ \delta /2, i \delta /2 \} $.   The generated state has a fidelity of $ F = 1 $ with respect to the state (\ref{catseq}) with $ \phi = 0 $. For (b), we use the same parameters as Fig. \ref{fig4}(e) ($M=3, t_{\max}=0.8, \theta = 0 $) to generate the initial squeezed vacuum state, before applying a further $M =2 $ measurements in (\ref{gkpstate}). For (c), we use the same parameters as Fig. \ref{fig2}(a) ($M=11, t_{\max}=2, \theta = 0 $) to generate the initial squeezed vacuum state, before applying a further $M =2 $ measurements in (\ref{gkpstate}). Fidelities comparison to the approximate GKP state (\ref{approxgkp1l}) after optimization of $ r, \sigma_{\text{env}} $ are shown in the figures. Fock state truncation of $ n_{\text{cut}} = 201 $ is used in (a)(b)  and $ n_{\text{cut}} = 301 $ for (c).  A lattice spacing $ \delta = 2 \sqrt{\pi} $ is used for all calculations.  
\label{fig6}   }  
\end{figure}

\subsection{GKP states}

Another paradigmatic family of states we can prepare with our protocol are Gotesmann-Knill-Preskill (GKP) states, which largely used for encoding logical qubit states into bosonic modes.
The GKP codewords are defined as  \cite{gottesman2001encoding}
\begin{subequations}\label{gkpdefs}
\begin{align}
|0_L \rangle & \propto \sum_{j=-\infty}^\infty  |x_0 = 2j \sqrt{\pi} \rangle \\
|1_L \rangle & \propto \sum_{j=-\infty}^\infty  |x_0 = (2j+1) \sqrt{\pi} \rangle .
\end{align}
\end{subequations}
These are a superposition of infinitely $ x $ squeezed states with lattice spacing $ \delta = 2 \sqrt{\pi} $.  Since ideal GKP codewords have infinite energy, experimental realization rely of their approximation using finitely squeezed states and an overall Gaussian envelope in the sums appearing in Eqs.~(\ref{gkpdefs}) \cite{weigand2018generating}.  
%For this reason, Eq.~(\ref{gkpdefs}) is defined up to a proportionality constant to account for normalization. 

Following this observation, approximated GKP states may be prepared by combining our protocols for generating squeezed states and superpositions of coherent states.  
The protocol we propose is shown in Fig. \ref{fig1}(c). First, an $ x $-squeezed vacuum state is generated using the measurement sequence ({\color{red}\ref{bigprojdef}}).  
Then, it is displaced by half a period $ \delta /2 = \sqrt{\pi} $ and the projection $ P_+ $ is performed through parity measurement. 
Let us remember that, according to Eq.~(\ref{catparity}), this measurement creates a point reflection of a coherent states.  
Since an arbitrary state $ | \psi \rangle $ can be decomposed in terms of coherent states, this creates a superposition according to 
\begin{align}
P_+ | \psi \rangle &  = \frac{1}{2\pi} \int d^2  \alpha  \langle \alpha | \psi \rangle ( | \alpha \rangle  +  |-  \alpha \rangle ) \nonumber \\
& = \frac{1}{2} ( | \psi \rangle + | - \psi \rangle ) ,
\end{align}  
where we applied the identity operator  $ I = \int d^2  \alpha |\alpha \rangle \langle \alpha |/ \pi $   and defined
\begin{align}
 | - \psi \rangle :=  \frac{1}{\pi} \int d^2  \alpha  \langle \alpha | \psi \rangle  |-  \alpha \rangle
\end{align} 
as the point reflected state.  
Hence the parity operator creates a superposition of the original state and its point reflected state.  
Since the $ x $-squeezed vacuum state is symmetric around the origin, applying the sequence given in Fig. \ref{fig1}(c) approximately produces the GKP codeword $ |1_L \rangle $.  
The state can be written as
\begin{align}
|{\cal G} (\vec{t}, \vec{\alpha} )  \rangle :=  \left[  \prod_{m=1}^M P_+ D\left(\frac{ m \delta}{2} \right) \right]  | {\cal P}_0 ( \vec{t} ) \rangle  \approx | 1_L \rangle ,
\label{gkpstate}
\end{align}
where the initial $ x $-squeezed vacuum state is generated using Eq.~(\ref{claim}).  

Figure \ref{fig6}(b)(c) shows the Wigner function for the states resulting from Eq.~(\ref{gkpstate}). 
The two panels correspond to states obtained by starting the protocol from either the squeezed state prepared with $M =3 $ measurements (Fig. \ref{fig6}(b)), or with $ M=11$ measurements (Fig. \ref{fig6}(c)), respectively.  
We see that the pattern characteristic of GKP states is well reproduced by our approach (c.f. Ref. \cite{lee2025photonic}).  The Wigner function for the $ |1_L \rangle $ GKP state is characterized by a positive distribution at $ \langle x \rangle = (2j+1)\sqrt{\pi} = (\pm \delta/2, \pm 3 \delta/2, \dots) $. Between these, there is a pattern of positive and negative peaks at $ \langle x \rangle = 2j\sqrt{\pi} = (0, \pm \delta , \dots) $, matching the patterns seen in Fig. \ref{fig6}(b)(c).  The higher level of squeezing for the $ M = 11 $ measurement case produces a wider distribution in the $ p $-direction as expected, but also greater squeezing in the $ x $-direction of the lattice points in the Wigner distribution.   We compare the resulting state with the approximate GKP state with a Gaussian envelope, defined as
\begin{align}
    |1_L, \xi, \sigma \rangle \propto  \sum_{j=-\infty}^\infty e^{ - \frac{((2j+1)\sqrt{\pi} )^2}{2 \sigma_{\text{env}}^2} } D ((2j+1)\sqrt{\pi})  S(r) |0 \rangle  .
    \label{approxgkp1l}
\end{align}
up to a normalization factor. After optimizing the parameters $ r $ and $ \sigma_{\text{env}} $ we find that the resulting states have a fidelity $ F \approx 0.99$.

\subsection{Effect of losses}

We analyze the sensitivity of the cat and GKP state generation to loss. We use the same procedure as shown in (\ref{measurementcycle}),
where loss is applied after each displaced parity measurement. The results are shown in      
Fig.~\ref{fig7}. The cat state preparation retains a comparatively large success probability $p_{\text{suc}}$ over the parameter range considered, with only a weak dependence on loss.
For the GKP state, the success probability decreases with increasing displacement amplitude and becomes more sensitive to loss for larger displacements.  In particular, the longer $M=11$ sequence becomes more strongly affected, with success probabilities dropping with added loss.  The Wigner functions for the cat state show generally similar features under weak loss, with a reduction in contrast of the negative regions.  For the GKP state, the fine features are still visible but again with a reduced amplitude.  The longer $M=11$ GKP sequence accumulates more photon loss and results in a considerably  lowered fidelity.   

Generally, we see that loss tends to more strongly affect the cat and GKP state generation than the squeezed states.  This is expected due to the higher sensitivity of these states to decoherence, as they possess highly non-classical features.  The nature of the state generation is also different in the case of cat and GKP states, where the measurements induce dynamical point reflections in phase space, rather than enforcing a common eigenstate as in the squeezed state case.

\begin{figure}[t]
\includegraphics[width=\linewidth]{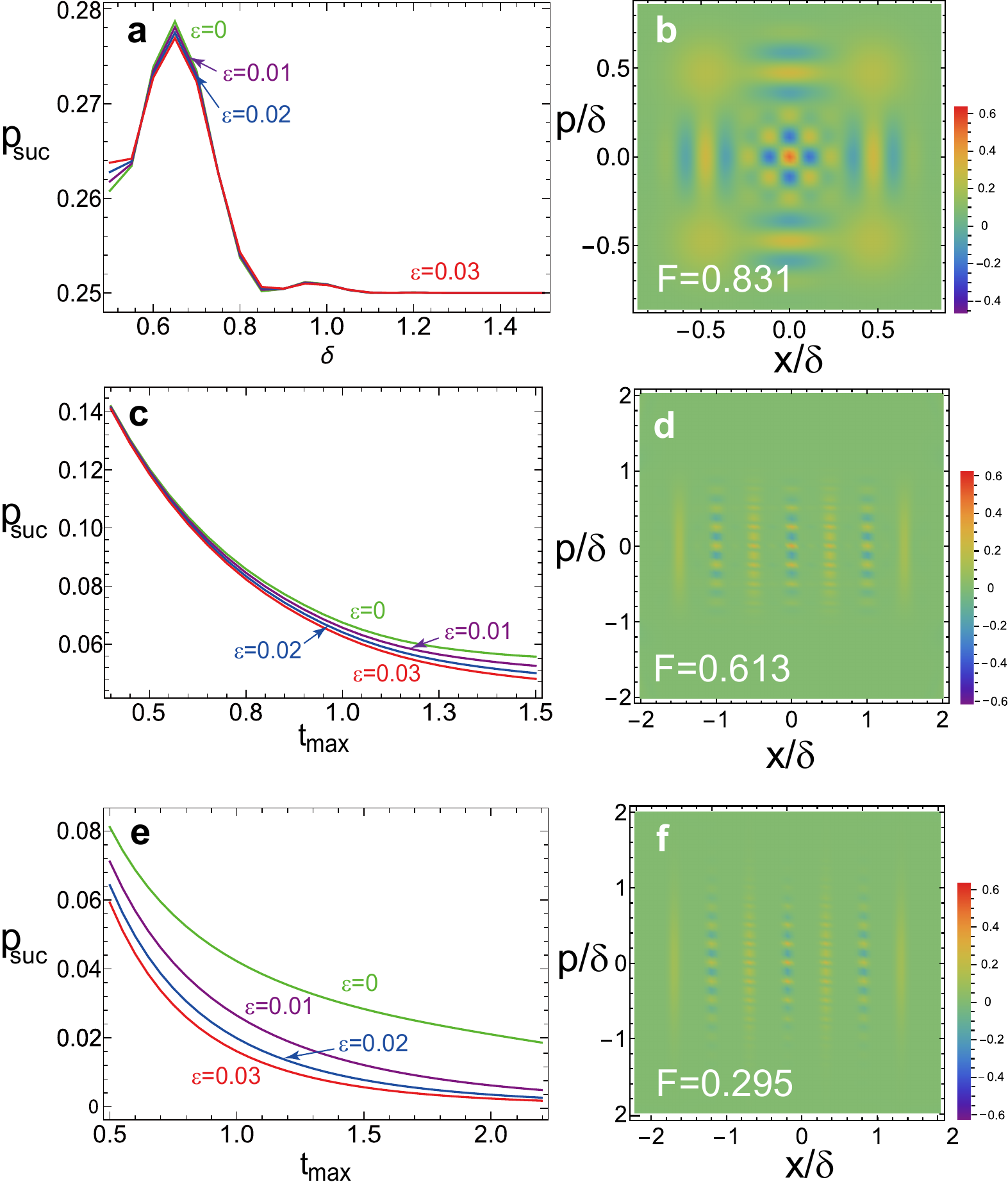}
\caption{Photon loss sensitivity for cat and GKP state preparations of Fig.~\ref{fig6}.  
Panels (a)(c)(e) show the total post-selection probability for the  cat state in Fig.~\ref{fig6}(a), the $M=3$ GKP state in Fig.~\ref{fig6}(b), and the $M=11$ GKP state in Fig.~\ref{fig6}(c), respectively. 
Panels (b)(d)(f) show the associated Wigner functions at $ \epsilon = 0.02 $, which is the lossy counterpart of Fig.~\ref{fig6}(a)(b)(c), respectively.
}

\label{fig7}   
\end{figure}

\section{Universal state generation}

In this section, we show that the combination of the measurements (\ref{measurementops}) and displacement operations has the ability to generate an arbitrary bosonic state. The aim here is to show that it is possible, in principle, to generate an arbitrary state, rather than a practical scheme for generating a given state. As before, postselection of measurement operators (\ref{measurementops}) will be allowed for any non-zero success probability.

\subsection{Universal gate sets}

Our main strategy will be to adapt the universality proof for continuous variable modes \cite{lloyd1999quantum} to our case.  The crux of the argument in Ref. \cite{lloyd1999quantum} is that given that if one has available the Hamiltonians $ H_1, H_2 $, then one may also implement the Hamiltonian $  H_3 = i [H_1, H_2] $ through the relation
\begin{align}
e^{i H_1 \epsilon} e^{i H_2 \epsilon} e^{-i H_1 \epsilon} e^{- i H_2 \epsilon} = e^{[H_1, H_2] \epsilon^2} + O(\epsilon^3) ,
\label{bch}
\end{align}
which follows from the Baker-Campbell-Hausdorff formula.  
Reusing the relation (\ref{bch}) with $ H_1 $ and $H_3 $, one may generate further Hamiltonians $ H_4 = i [H_1, H_3] $, for example. This allows one to build up a gate set, which may then be used to generate unitary evolutions of the form
\begin{align}
e^{i \sum_m w_m H_m }  \approx  \left( \prod_m e^{i w_m H_m/M } \right)^M
\label{genunitary}
\end{align}
where we used the Lie-Trotter formula.  Universality follows if the set of Hamiltonians $ \{ H_m \} $ generated by successive commutations span the full operator space.  In the case of bosonic modes, displacements, phase rotations, and a nonlinear Hamiltonian such as the Kerr Hamiltonian $ (a^\dagger a)^2 $ is sufficient for universality \cite{lloyd1999quantum}. 

To show universality in our case, we will make use of the imaginary time version of (\ref{bch}), which reads
\begin{align}
e^{- H_1 \epsilon} e^{- H_2 \epsilon} e^{+H_1 \epsilon} e^{+H_2 \epsilon} = e^{[H_1, H_2] \epsilon^2} + O(\epsilon^3) .
\label{bchimag}
\end{align}
We also use the mixed case where $H_1 $ is evolved under imaginary time evolution, while $ H_2 $ is unitarily evolved:
\begin{align}
e^{- H_1 \epsilon} e^{- i H_2 \epsilon} e^{+H_1 \epsilon} e^{+i H_2 \epsilon} = e^{i [H_1, H_2] \epsilon^2} + O(\epsilon^3) .
\label{bchimagunit}
\end{align}
This reveals an interesting and important fact.  Due to the fact that $  H_3 = i [H_1, H_2] $ is Hermitian $ H_3^\dagger = H_3 $, then two imaginary time gates as given in (\ref{bchimag}) can be combined to give a {\it unitary} gate, i.e. real time evolution.  Combining an imaginary time gate with a unitary gate as given in (\ref{bchimagunit}) produces an {\it imaginary time} gate.  

Our aim will be to construct the imaginary time version of (\ref{genunitary}), which reads
\begin{align}
e^{-\sum_m w_m H_m }  \approx  \left( \prod_m e^{- w_m H_m/M } \right)^M  .
\label{genimagevol}
\end{align}
Here, the gates on the right hand side are imaginary time gates, hence they are constructed using (\ref{bchimagunit}).  As with the unitary case, as long as the set of  Hamiltonians ${\cal H_\text{imag} } = \{ H_m \} $ that can be generated span the full operator space, one can construct an arbitrary state since 
\[
 \lim_{t \rightarrow \infty}  \frac{e^{- H_\psi t}\ket{\psi_0}}
{\norm{e^{- H_\psi t}\ket{\psi_0}}}
= 
\ket{\psi}
\]
%
% \begin{align}
%    e^{-H_{\psi} t  }  | \psi_0 \rangle \propto | \psi \rangle,
% \end{align}
%
where $ |\psi \rangle $ is the non-degenerate ground state of the Hamiltonian 
\begin{align}
    H_\psi = \sum_m w_m H_m .
\end{align} 
Here, the initial state $ | \psi_0 \rangle  $ must have a non-zero overlap with the target state $ | \psi \rangle $.   Specifically, we may choose
\[
H_\psi
=
I-\ket{\psi}\bra{\psi}.
\]
This Hamiltonian has ground state $\ket{\psi}$ with eigenvalue $0$, while every
state orthogonal to $\ket{\psi}$ has eigenvalue $1$.

\subsection{Measurement-based imaginary time gates}
\label{sec:measimag}

We now discuss how to make the individual imaginary time gates on the right hand side of (\ref{genimagevol}). To this end, we use the measurement operators (\ref{measurementops}) to show that they approximate several elementary imaginary time gates.  A similar approach was discussed in Ref. \cite{mao2023measurement}. 

First, taking 
$  \tau = \pm \epsilon, 0 < \epsilon \ll 1 , \phi =  \pi/4 $ for the $ {\cal M}_- $ outcome in (\ref{measurementops}), we have
\begin{align}
{\cal M}_- (\pm \epsilon ,\pi/4) & =  \frac{1}{\sqrt{2}} ( \cos( \hat{n} \epsilon) \pm  \sin ( \hat{n} \epsilon) ) \\
& \approx  \frac{1}{\sqrt{2}} (  I  \pm  \hat{n} \epsilon ) \approx   \frac{1}{\sqrt{2}} e^{\pm \hat{n} \epsilon }, 
\label{n1gate}
\end{align}
where $ \hat{n} = a^\dagger a$. This gives an imaginary time gate with Hamiltonian $ \propto  \hat{n} $.  In (\ref{n1gate}), we used a linear expansion of the trigonometric functions to obtain an approximation to the imaginary time gate.  The exponential function may be further improved using more measurements using the methods of Appendix \ref{app:interp} at the cost of a lower postselection probability.    

A nonlinear imaginary time Hamiltonians with a negative coefficient may be produced using $  \tau = \epsilon, 0 < \epsilon < \ll 1, \phi =  \pi/2 $ for the $ {\cal M}_- $ outcome in (\ref{measurementops}), 
\begin{align}
{\cal M}_-  (\epsilon ,\pi/2) & =  \cos (\hat{n} \epsilon)  \\
& \approx I - \frac{ \hat{n}^2}{2} \epsilon^2 \approx  e^{  - \hat{n}^2 \epsilon^2/2 } .
\label{n2gate}
\end{align}
Again, the approximation to the imaginary time gate may be further improved using multiple measurements as shown in Appendix \ref{app:gaussianneg}.  It is also possible to obtain a positive exponent using multiple measurements as described in Appendix \ref{app:interp} and \ref{app:gaussianpos}.

Further imaginary time gates may be produced by combining the imaginary time gates with unitary displacement operations.  Taking $ H_1 = \hat{n} $  and $ H_2 = \alpha a+ \alpha^* a^\dagger $ as the displacement Hamiltonian in (\ref{bchimagunit}), we have
\begin{align}
i[H_1, H_2] = i[\hat{n}, \alpha a+ \alpha^* a^\dagger  ] = -i \alpha a+ i \alpha^* a^\dagger .
\end{align}
This produces the imaginary time gate $ e^{-i \alpha a+ i \alpha^* a^\dagger} $.  

Using $ H_1 = \hat{n}^2 $ and $ H_2 = \alpha a+ \alpha^* a^\dagger $ in (\ref{bchimagunit}) yields
\begin{align}
i[H_1, H_2] & =i [\hat{n}^2, \alpha a+ \alpha^* a^\dagger  ] \\
& = 2\hat{n} (-i \alpha a+ i  \alpha^* a^\dagger ) + i (\alpha a+ \alpha^* a^\dagger) .
\label{thirdorderham}
\end{align}
This produces an operator that is cubic in bosonic operators. Here, since $ \alpha $ is an arbitrary complex number, the sign is freely choosable.  This gives another route to obtain a nonlinear imaginary time operator using only (\ref{n2gate}) and displacement operations, but with a freely choosable sign.  

Commuting the result of (\ref{thirdorderham}) with the displacement Hamiltonian again produces
\begin{align}
i [i[H_1, H_2], H_2] = -2 (  \alpha^2 a^2 + (\alpha^* a^\dagger)^2) +  2 |\alpha|^2 ( 4 \hat{n} + 1 )  .
\end{align}
Subtracting off the terms proportional to $ \hat{n} $, this produces an imaginary time gate counterpart of the squeezing Hamiltonian.  

So far, the above procedure has produced imaginary time gates corresponding to the Hamiltonians
\begin{align}
{\cal H_\text{imag} } = & \{ -i \alpha a+ i \alpha^* a^\dagger ,  \alpha^2 a^2 + (\alpha^* a^\dagger)^2,  \hat{n}, \nonumber \\
& \hat{n} (- i\alpha a+  i \alpha^* a^\dagger ), \hat{n}^2 \}
\label{imaghams}
\end{align}
Commuting these imaginary time Hamiltonians with the displacement Hamiltonian $ H_2 = \alpha a+ \alpha^* a^\dagger $ always reduces the order of the Hamiltonian in terms of bosonic operators.  Thus to proceed further we require a real time Hamiltonian of higher order.  This can be produced by combining two imaginary time gates and using (\ref{bchimag}).  For example, using $ H_1 = \hat{n}^2 $ and $ H_2 =  -i \alpha a+ i \alpha^* a^\dagger $, we obtain
\begin{align}
i[H_1, H_2] & = i [ \hat{n}^2, -i \alpha a+ i \alpha^* a^\dagger] \\
& = - 2 \hat{n} (\alpha a + \alpha^* a^\dagger) + \alpha a - \alpha^* a^\dagger  ,
\end{align}
which produces a {\it unitary} evolution.  This gives an additional real time gate, which is nonlinear.  Thus the obtained gate set up to this point is
\begin{align}
H_{\text{real}} = & \{ \alpha a+  \alpha^* a^\dagger , \hat{n} (\alpha a+ \alpha^* a^\dagger ) \} .  
\label{realhams}
\end{align}
Further combining the imaginary time Hamiltonians (\ref{imaghams}) with the unitary Hamiltonians (\ref{realhams}), we may produce bosonic imaginary time Hamiltonians of any order. Similarly to the arguments in Ref. \cite{lloyd1999quantum}, the nonlinear (real time) Hamiltonian in (\ref{realhams}) increases the order of the imaginary time Hamiltonians (\ref{imaghams}) using (\ref{bchimagunit}).  This shows that it is possible to produce an arbitrary imaginary time evolution, hence an arbitrary state using the combination of measurements (\ref{measurementops}) and unitary displacement operations.

\section{Summary and conclusions}

We have explored the use of  displacement operations and postselected dispersive measurements, in particular parity measurements,  for the preparation of squeezed vacuum states, multi-component cat states, and GKP states.  
The fundamental principle that allows for the generation of squeezed states was shown in (\ref{measseq}), namely the fact that an infinitely squeezed state is an eigenstate of the displaced parity operator.  
%This allows for a rapid convergence of the state towards the squeezed state.  
In the ideal case, high squeezing levels of up to 21.5 dB can be obtained with a modest number of measurements ($N=11$) and reasonably high success probabilities ($p_{\text{suc}} = 8.5\% $). 
Even for a smaller number of measurements ($N=3$), 8.9 dB of squeezing can be obtained with a high success probability ($p_{\text{suc}} = 32 \% $).   
These results are based on a parity measurement sequence we have found, which results to good approximation in squeezed states, see Fig. \ref{fig1}(a).  
Under loss, squeezing levels are reduced, but still remain considerably large under realistic experimental conditions.  

Besides squeezed states, we have shown that our approach can be generalized to the preparation of other types of states, such as multi-component cat states.  
The key insight is that parity measurement can be considered as an operation that creates a superposition of the original and its point-reflected state.  
Using this basic idea, lattices of coherent states can be generated, which form the basis for the preparation of GKP states.
Under loss, these states show a higher sensitivity than the squeezed states.  We attribute this mainly due to the fundamentally more fragile nature of these states, and the more ``dynamical'' way that these states are generated.  For squeezed states, the projection sequence is more convergent process, where each measurement stabilizes a common eigenstate. 

More generally, we showed that it is possible to generate an arbitrary state using dispersive measurements (\ref{measurementops}) and displacement operations. The key feature that is used to show this is the nonlinearity of measurements, which produces an effective imaginary time Kerr Hamiltonian evolution.  Under postselection, the measurement operators mimic imaginary time gates \cite{mao2023measurement}.  Using a sequence of measurements it is possible to engineer number basis filters of an arbitrary form.  This suggests that our scheme can be used in a similar way to quantum control methods such as Gradient Ascent Pulse Engineering (GRAPE) \cite{khaneja2005optimal} and Selective Number dependent Arbitrary Phase (SNAP) \cite{heeres2015snap} where states are engineered using an optimized gate sequence.  In particular, there is a strong similarity with SNAP which works with a combination of phase rotations and displacements. One major difference to these schemes is that these are deterministic, as all gates are unitary.   In our case, instead of phase rotations, parity measurements are utilzed  as the non-linear resource.  
% We have compared our method to SNAP in the context of producing a squeezed state and have found favorable results (see Supplementary Information).  However, due to the different operations that are used, a direct comparison is difficult without quantify the resources required to make the operations. 
An interesting avenue of future work would be use the techniques developed in SNAP and apply them to our scheme.  

%In our case, as postselection is used, one would require maximizing not only the final state fidelity but also the success probability.  Further investigation of such methods is left as future work. 

The main drawback of our technique is that postselection is required, which results in potentially a low success probability.   
A way to circumvent this is to use imaginary time evolution methods \cite{mao2023measurement}, which can prepare any eigenstate of a given Hamiltonian and has been shown to produce a variety of states \cite{kondappan2023imaginary,chen2024efficient}.   
The idea here is to simulate a measurement in the energy eigenbasis of the Hamiltonian and then perform adaptive unitary operations such as to ensure that convergence in the desired ground state is obtained. 
While more complex than postselection, it has the advantage that higher efficiencies can be obtained.

% {\color{red} Existing protocols for measurement-based preparation of squeezed states rely on quadrature measurements, making them less compatible for realization} on cavity/circuit-QED systems, trapped ions and atoms.
% In these platforms, in fact, the bosonic mode is read-out through a two-level degree of freedom, making more natural to perform displaced parity measurements.
% Therefore, our work opens up new possibilities for state preparation in microwave and optical cavities, as well as motional states of trapped ions, atoms and solid-state oscillators.

\section*{Declarations}

\subsection{Author Contributions}

Z.L. and S.L. contributed equally to this work. Z.L. and S.L. performed the calculations, T. B. and M.F. formulated the theoretical concept. 
J.F. assisted with numerical analysis and validation. K.Z., T. M., and T.B. performed analytical calculations. 
M.F.,V.I., and T.B. supervised the project. 
Z.L., S.L., T.B., M. F.  wrote the paper. All authors discussed the results and reviewed the manuscript.

\subsection{Competing Interests}

The authors declare no competing interests.

\begin{acknowledgments}
This work is supported by the SMEC Scientific Research Innovation Project (2023ZKZD55); the National Natural Science Foundation of China (92576102); the Science and Technology Commission of Shanghai Municipality (22ZR1444600); the NYU Shanghai Boost Fund; the China Foreign Experts Program (G2021013002L); the NYU-ECNU Institute of Physics at NYU Shanghai; the NYU Shanghai Major-Grants Seed Fund; and Tamkeen under the NYU Abu Dhabi Research Institute grant CG008. 
M.F. was supported by the Swiss National Science Foundation Ambizione Grant No. 208886, and by The Branco Weiss Fellowship -- Society in Science, administered by the ETH Z\"{u}rich.
\end{acknowledgments}

\appendix

\section{Hamiltonian for infinitely squeezed state}
\label{sec:hamsqueezed}

An alternative formulation of the infinitely squeezed state is as the ground state of the Hamiltonian 
\begin{align}
    H = -\sum_{n=1}^N  P_+ ( \alpha = it_n) .
    \label{ham}
\end{align}
It is guaranteed that $ |x_0 \rangle $ is the ground state given that the eigenspectrum of (\ref{ham}) is within the range $ [-N,N] $ since each term within the sum takes a value $ [-1,1] $ for any expectation value.  The eigenvalue of $ |x_0 \rangle $  is $ -N $ which shows that it is the minimum value within the allowed range, showing it must be a ground state.

\section{Interpolation with Real Rooted Polynomials}
\label{app:interp}

In this section, we show that the imaginary time operators discussed in Sec. \ref{sec:measimag} can be improved further beyond the low-order expansions used in (\ref{n1gate}) and (\ref{n2gate}).  

Our approach will be to use a sequence of measurement operators (\ref{measurementops}), which take a general form
\begin{align}
{\cal M} (\vec{\tau}, \vec{\phi}) & =  \prod_{m=1}^M {\cal M}_- ( \tau_m, \phi_m) \label{seqmeas} \\
& = \sum_n \left( \prod_{m=1}^M \sin (n \tau_m + \phi_m) \right) | n \rangle \langle n| 
\end{align}
where we defined
\begin{align}
\vec{\tau} & =  (\tau_1, \tau_2, \dots, \tau_M) \\
\vec{\phi} & =  (\phi_1, \phi_2, \dots, \phi_M)  .
\end{align}
Eq. (\ref{seqmeas}) corresponds to a sequence of measurements postselected all on the $ {\cal M}_- $ outcome, each with an adjustable interaction time $ \tau_m $ and phase offset $ \phi_m $ on the $m$th measurement.  There are a total of $ M $ measurements.  

We work within a truncated Fock space such that $ n \in \{0, \dots, N\} $, and assume that $ |\tau_m| N , |\phi_m| \ll 1 $, such that  we may approximate
\begin{equation}
{\cal M} (\vec{\tau}, \vec{\phi}) 
  \approx \sum_{n=0}^{N} \left( \sum_{m} \tau_{m} \right)
     \prod_{m=1}^{M} (n + x_{m}) \lvert n \rangle \langle n |  ,
     \label{approxmeasureseq}
\end{equation}
where $ x_{m} = \phi_{m}/\tau_{m} $.  Our aim will be to try to construct a real-valued 
diagonal  filter $\mathcal{F} = \sum_{n=0}^{N} f_{n} \lvert n \rangle \langle n \rvert$ using (\ref{approxmeasureseq}).  The product in (\ref{approxmeasureseq}) takes the form of a polynomial in $ n $ with roots at $ -x_m $.  It is however restricted as all the roots are real, not complex, from the Fudamental Theorem of Algebra. Thus it is not possible to write all possible polynomials in $ n $ using the form (\ref{approxmeasureseq}).  However, since $n $ is discrete variable, there is still a possibility that there exists a polynomial to fit an arbitrary set of values $ f_n $.  

That is, we can use (\ref{approxmeasureseq}) to construct a general filter $\hat{\mathcal{F}} $ if there exists a real rooted polynomial function $p(x) = c \prod_{m=1}^{M} (x +  x_{m})$ such that $p(n) = f_{n}$ for all $n \in \{0, \dots, N\}$. This is proven below.  

\begin{theorem}
Let
\[
(x_1,y_1),\dots,(x_n,y_n)\in\mathbb R^2,
\qquad x_1<x_2<\cdots<x_n.
\]
Then there exists a real polynomial \(p\) of degree \(n\), all of whose roots are real, such that
\[
p(x_j)=y_j,\qquad j=1,\dots,n.
\]
Equivalently,
\[
p(x)=c\prod_{j=1}^n (x-r_j)
\]
for some \(c\in\mathbb R\setminus\{0\}\) and real numbers \(r_1,\dots,r_n\).
\end{theorem}

\begin{proof}
By the Lagrange interpolation theorem, there exists a unique polynomial
\(\mathcal L\) of degree at most \(n-1\) such that
\[
\mathcal L(x_j)=y_j,\qquad j=1,\dots,n.
\]
Define
\[
F(x)=\prod_{j=1}^n (x-x_j).
\]
Then \(F(x_j)=0\) for every \(j\). Hence for any \(\lambda\in\mathbb R\), the polynomial
\[
p_\lambda(x)=\mathcal L(x)+\lambda F(x)
\]
satisfies
\[
p_\lambda(x_j)=\mathcal L(x_j)+\lambda F(x_j)=y_j.
\]

It remains to choose \(\lambda\) so that \(p_\lambda\) has \(n\) real roots. Since the roots \(x_1,\dots,x_n\) of \(F\) are simple, \(F\) changes sign at each \(x_j\). Choose \(\delta_j>0\) sufficiently small so that the intervals
\[
I_j=(x_j-\delta_j,x_j+\delta_j)
\]
are pairwise disjoint and contain no other zero of \(F\). Then
\[
F(x_j-\delta_j)F(x_j+\delta_j)<0.
\]

Moreover, since none of the endpoints \(x_j\pm\delta_j\) is a zero of \(F\), we have
\[
|F(x_j\pm\delta_j)|>0.
\]
Choose \(\lambda>0\) sufficiently large such that, for every \(j\),
\[
\lambda |F(x_j\pm\delta_j)|
>
|\mathcal L(x_j\pm\delta_j)|.
\]
Then at each endpoint, the term \(\lambda F(x)\) dominates \(\mathcal L(x)\). Therefore
\[
\operatorname{sgn}(p_\lambda(x_j\pm\delta_j))
=
\operatorname{sgn}(F(x_j\pm\delta_j)).
\]
Since \(F\) has opposite signs at \(x_j-\delta_j\) and \(x_j+\delta_j\), we get
\[
p_\lambda(x_j-\delta_j)p_\lambda(x_j+\delta_j)<0.
\]
By the intermediate value theorem, \(p_\lambda\) has at least one real root in each interval \(I_j\).

Because the intervals \(I_1,\dots,I_n\) are pairwise disjoint, \(p_\lambda\) has at least \(n\) distinct real roots. Also,
\[
p_\lambda(x)=\mathcal L(x)+\lambda F(x)
\]
has degree exactly \(n\), since \(\lambda\neq 0\) and \(F\) has degree \(n\). Hence \(p_\lambda\) has exactly \(n\) roots over \(\mathbb C\), counted with multiplicity. Since it already has \(n\) distinct real roots, all its roots are real.

Therefore
\[
p_\lambda(x)=c\prod_{j=1}^n (x-r_j)
\]
with \(r_j\in\mathbb R\), and \(p_\lambda(x_j)=y_j\) for all \(j=1,\dots,n\).
\end{proof}

% For a sufficiently large number of measurements $ M $, the measurement sequence can reproduce an arbitrary discrete real valued function, up to a constant factor (see Appendix \ref{app:interp}).  This allows us to construct individual imaginary time gates.  A similar approach was used in Ref. \cite{mao2023measurement} to produce imaginary time gates. 

\section{Measurement-based $\hat{n}^2 $ imaginary time gates}

In this section we show further details regarding the implementation of the nonlinear imaginary time gate (\ref{n2gate}).

\subsection{Negative exponent}
\label{app:gaussianneg}

Here, we show an improved method of obtaining the approximation  (\ref{n2gate}).  

% First we define the displaced version of the $ {\cal M}_+ $ operator defined in (\ref{measurementops}), setting $ \phi = 0 $
% %
% \[
% {\cal M}_\alpha (\tau)
% =
% D(\alpha)\cos( \hat{n} \tau)D^\dagger(\alpha)
% =
% \cos(\hat{n}_\alpha\tau),
% \]
% where
% \[
% \hat{n}_\alpha
% =
% D(\alpha) \hat{n} D^\dagger(\alpha)
% =
% (a^\dagger-\alpha^*)(a-\alpha).
% \]
% Equivalently,
% \[
% \hat{n}_\alpha
% =
% \hat{n}-\alpha a^\dagger-\alpha^*a+|\alpha|^2.
% \]

The measurement operator $ {\cal M}_+ $ defined in (\ref{measurementops}), setting $ \phi = 0 $ for small $\tau$ is 
\[
{\cal M}_+(\tau, 0)  = \cos(\hat{n} \tau)
=
1-\frac{\tau^2}{2}\hat{n}^2+O(\tau^4).
\]
Choose
\[
\tau=\sqrt{\frac{2t}{K}}.
\]
Then
\[
{\cal M}_+ (\tau,0)^K
=
\left[
1-\frac{t}{K} \hat{n}^2
+
O\left(\frac{1}{K^2}\right)
\right]^K.
\]
Taking the limit $K\to\infty$, we obtain
\[
\lim_{K\to\infty}
{\cal M}_+ \left(\sqrt{\frac{2t}{K}},0 \right)^K
=
e^{-t \hat{n}^2}, 
\]
as claimed. 

% This is a nonunitary gate. It is not generated by a Hamiltonian evolution
% $e^{-iHt}$, but by an imaginary-time evolution.
% To use a BCH commutator construction, we need access to both
% \[
% e^{-t \hat{n}_\alpha^2}
% \qquad
% \text{and}
% \qquad
% e^{+t \hat{n}_\alpha^2}.
% \]
% The negative sign follows directly from repeated weak measurements:
% \[
% \left[
% \cos\left(\hat{n}_\alpha\sqrt{\frac{2t}{K}}\right)
% \right]^K
% \to
% e^{-t \hat{n}_\alpha^2}.
% \]

\subsection{Positive exponent}
\label{app:gaussianpos}

In Appendix \ref{app:interp} we have already shown that it is possible to use a measurement sequence (\ref{approxmeasureseq}) to approximate an arbitrary filter.  Hence it follows that one can  create an imaginary time gate $ \propto e^{\hat{n}^2 t} $.  However, as this operator is unbounded we must analyze under what conditions the operator can be implemented.

For generality we work with the displaced version of the $ e^{\hat{n}^2 t} $ operator 
\[
e^{+ \hat{n}_\alpha^2 t }
=
D(\alpha) e^{\hat{n}^2 t}  D^\dagger(\alpha) ,
\]
where
\[
\hat{n}_\alpha
=
D(\alpha) \hat{n} D^\dagger(\alpha)
=
(a^\dagger-\alpha^*)(a-\alpha).
\]
Equivalently,
\[
\hat{n}_\alpha
=
\hat{n}-\alpha a^\dagger-\alpha^*a+|\alpha|^2.
\]

Let
\[
\mathcal H_N
=
\mathrm{span}\{\ket0,\ket1,\dots,\ket N\}.
\]
Assume that on every finite cutoff $\mathcal H_N$, the protocol can implement,
up to an irrelevant scalar normalization,
\[
e^{+t \hat{n}_\alpha^2}
\]
This assumption is reasonable only on a finite cutoff, because $e^{+t n^2}$ is
unbounded on the full infinite-dimensional Hilbert space.

On this finite-dimensional space, since $\hat{n}_\alpha^2$ is Hermitian on $\mathcal H_N$, it has a spectral decomposition
\begin{align}
    \hat{n}_\alpha^2
    =
    \sum_j \lambda_j \ket{\lambda_j}\bra{\lambda_j},
\end{align}
where $\lambda_j\in\mathbb R$. Define
\begin{align}
    \lambda_{\max}(\alpha)
    =
    \max_j \lambda_j .
\end{align}
Then
\begin{align}
    \lambda_j
    \leq
    \lambda_{\max}(\alpha)
\end{align}
for every $j$. Hence
\begin{align}
    \lambda_{\max}(\alpha)I-\hat{n}_\alpha^2
    &=
    \sum_j
    \left(
    \lambda_{\max}(\alpha)-\lambda_j
    \right)
    \ket{\lambda_j}\bra{\lambda_j}.
\end{align}
Because
\begin{align}
    \lambda_{\max}(\alpha)-\lambda_j\geq 0,
\end{align}
all eigenvalues of $\lambda_{\max}(\alpha)I-\hat{n}_\alpha^2$ are nonnegative.
Therefore,
\begin{align}
    \lambda_{\max}(\alpha)I-\hat{n}_\alpha^2
    \succeq 0.
\end{align}
Equivalently, for any $\ket{\psi}\in\mathcal H_N$,
\begin{align}
    \bra{\psi}
    \left(
    \lambda_{\max}(\alpha)I-\hat{n}_\alpha^2
    \right)
    \ket{\psi}
    \geq 0.
\end{align}

Now suppose we have the filter generated by
this positive semidefinite complementary operator:
\begin{align}
    e^{-t\left(\lambda_{\max}(\alpha)I-\hat{n}_\alpha^2\right)}.
\end{align}
Then
\begin{align}
    e^{-t\left(\lambda_{\max}(\alpha)I-\hat{n}_\alpha^2\right)}
    &=
    e^{-t\lambda_{\max}(\alpha)}
    e^{+t\hat{n}_\alpha^2}.
\end{align}
The factor $e^{-t\lambda_{\max}(\alpha)}$ is a scalar normalization
factor. In a postselected imaginary-time protocol, such an overall scalar is
irrelevant because the state is normalized after the operation. Therefore,
on the finite cutoff space $\mathcal H_N$, access to
\begin{align}
    e^{-t\left(\lambda_{\max}(\alpha)I-\hat{n}_\alpha^2\right)}
\end{align}
is equivalent, up to normalization, to access to
\begin{align}
    e^{+t\hat{n}_\alpha^2}
    =
    e^{+t \hat{n}_\alpha^2|_{\mathcal H_N}}.
\end{align}

Thus, on a finite cutoff, the protocol has access to both directions
\begin{align}
    e^{-t \hat{n}_\alpha^2}
    \qquad\text{and}\qquad
    e^{+t \hat{n}_\alpha^2},
\end{align}
provided that the complementary suppressive filter
\begin{align}
    e^{-t\left(\lambda_{\max}(\alpha)I-\hat{n}_\alpha^2\right)}
\end{align}
can be implemented on $\mathcal H_N$.

\bibliography{paperrefs}

@article{schonbeck201713,
  title={13 d{B} squeezed vacuum states at 1550 nm from 12 mW external pump power at 775 nm},
  author={Sch{\"o}nbeck, Axel and Thies, Fabian and Schnabel, Roman},
  journal={Optics letters},
  volume={43},
  number={1},
  pages={110--113},
  year={2017},
  publisher={Optical Society of America},
  doi = {https://doi.org/10.1364/OL.43.000110}
}

@article{aspelmeyer2014cavity,
  title={Cavity optomechanics},
  author={Aspelmeyer, Markus and Kippenberg, Tobias J and Marquardt, Florian},
  journal={Reviews of Modern Physics},
  volume={86},
  number={4},
  pages={1391--1452},
  year={2014},
  publisher={APS},
  doi={https://doi.org/10.1103/RevModPhys.86.1391}
}

@article{wineland1992spin,
  title={Spin squeezing and reduced quantum noise in spectroscopy},
  author={Wineland, David J and Bollinger, John J and Itano, Wayne M and Moore, FL and Heinzen, Daniel J},
  journal={Physical Review A},
  volume={46},
  number={11},
  pages={R6797},
  year={1992},
  publisher={APS},
  doi={https://doi.org/10.1103/PhysRevA.46.R6797}
}

@article{slusher1985observation,
  title={Observation of squeezed states generated by four-wave mixing in an optical cavity},
  author={Slusher, R\_E and Hollberg, LW and Yurke, Bernard and Mertz, JC and Valley, JF},
  journal={Physical review letters},
  volume={55},
  number={22},
  pages={2409},
  year={1985},
  publisher={APS},
  doi={https://doi.org/10.1103/PhysRevLett.55.2409}
}

@article{eichler2011observation,
  title={Observation of two-mode squeezing in the microwave frequency domain},
  author={Eichler, Christopher and Bozyigit, Deniz and Lang, Christian and Baur, Martin and Steffen, Lars and Fink, Johannes M and Filipp, Stefan and Wallraff, Andreas},
  journal={Physical Review Letters},
  volume={107},
  number={11},
  pages={113601},
  year={2011},
  publisher={APS},
  doi ={https://doi.org/10.1103/PhysRevLett.107.113601}
}

@article{yurke1989observation,
  title={Observation of parametric amplification and deamplification in a Josephson parametric amplifier},
  author={Yurke, Bernard and Corruccini, LR and Kaminsky, PG and Rupp, LW and Smith, AD and Silver, AH and Simon, RW and Whittaker, EA},
  journal={Physical Review A},
  volume={39},
  number={5},
  pages={2519},
  year={1989},
  publisher={APS},
  doi= {https://doi.org/10.1103/PhysRevA.39.2519}
}

@inproceedings{jabir2024quantum,
  title={Quantum enhanced precision metrology for quantum networks},
  author={Jabir, MV and Dawkins, Riley and Sabines-Chesterking, J and Reddy, Dileep V and Lita, AE and Battou, A and Gerrits, Thomas},
  booktitle={Quantum 2.0},
  pages={QTh4C--5},
  year={2024},
  organization={Optica Publishing Group},
  doi={https://doi.org/10.1364/QUANTUM.2024.QTh4C.5}
}

@article{blais2021circuit,
  title={Circuit quantum electrodynamics},
  author={Blais, Alexandre and Grimsmo, Arne L and Girvin, Steven M and Wallraff, Andreas},
  journal={Reviews of Modern Physics},
  volume={93},
  number={2},
  pages={025005},
  year={2021},
  publisher={APS},
  doi={https://doi.org/10.1103/RevModPhys.93.025005}

}

@article{wolf2019motional,
  title={Motional {F}ock states for quantum-enhanced amplitude and phase measurements with trapped ions},
  author={Wolf, Fabian and Shi, Chunyan and Heip, Jan C and Gessner, Manuel and Pezz{\`e}, Luca and Smerzi, Augusto and Schulte, Marius and Hammerer, Klemens and Schmidt, Piet O},
  journal={Nature communications},
  volume={10},
  number={1},
  pages={2929},
  year={2019},
  publisher={Nature Publishing Group UK London},
  doi={https://doi.org/10.1038/s41467-019-10576-4}
}

@article{dassonneville2021dissipative,
  title={Dissipative stabilization of squeezing beyond 3 d{B} in a microwave mode},
  author={Dassonneville, R and Assouly, R and Peronnin, T and Clerk, AA and Bienfait, A and Huard, B},
  journal={PRX Quantum},
  volume={2},
  number={2},
  pages={020323},
  year={2021},
  publisher={APS},
  doi={https://doi.org/10.1103/PRXQuantum.2.020323}
}

@article{li2011engineering,
  title={Engineering squeezed states of microwave radiation with circuit quantum electrodynamics},
  author={Li, Peng-Bo and Li, Fu-Li},
  journal={Physical Review A—Atomic, Molecular, and Optical Physics},
  volume={83},
  number={3},
  pages={035807},
  year={2011},
  publisher={APS},
  doi ={https://doi.org/10.1103/PhysRevA.83.035807}
}

@article{krasnok2024superconducting,
  title={Superconducting microwave cavities and qubits for quantum information systems},
  author={Krasnok, Alex and Dhakal, Pashupati and Fedorov, Arkady and Frigola, Pedro and Kelly, Michael and Kutsaev, Sergey},
  journal={Applied Physics Reviews},
  volume={11},
  number={1},
  year={2024},
  publisher={AIP Publishing},
  doi={https://doi.org/10.1063/5.0155213}
}

@article{jia2022determination,
  title={Determination of multimode motional quantum states in a trapped ion system},
  author={Jia, Zhubing and Wang, Ye and Zhang, Bichen and Whitlow, Jacob and Fang, Chao and Kim, Jungsang and Brown, Kenneth R},
  journal={Physical Review Letters},
  volume={129},
  number={10},
  pages={103602},
  year={2022},
  publisher={APS},
  doi={10.1103/physrevlett.129.103602}
}

@article{sutherland2021motional,
  title={Motional squeezing for trapped ion transport and separation},
  author={Sutherland, Robert T and Burd, SC and Slichter, DH and Libby, SB and Leibfried, Dietrich},
  journal={Physical review letters},
  volume={127},
  number={8},
  pages={083201},
  year={2021},
  publisher={APS},
  doi={https://doi.org/10.1103/PhysRevLett.127.083201}
}

@article{marti2024quantum,
  title={Quantum squeezing in a nonlinear mechanical oscillator},
  author={Marti, Stefano and von L{\"u}pke, Uwe and Joshi, Om and Yang, Yu and Bild, Marius and Omahen, Andraz and Chu, Yiwen and Fadel, Matteo},
  journal={Nature Physics},
  volume={20},
  number={9},
  pages={1448--1453},
  year={2024},
  publisher={Nature Publishing Group UK London},
  doi={https://doi.org/10.1038/s41567-024-02545-6}
}

@article{breitenbach1997measurement,
  title={Measurement of the quantum states of squeezed light},
  author={Breitenbach, Gerd and Schiller, S and Mlynek, J},
  journal={Nature},
  volume={387},
  number={6632},
  pages={471--475},
  year={1997},
  publisher={Nature Publishing Group UK London},
  doi    = {https://doi.org/10.1038/387471a0}
}

@article{banaee2008squeezed,
  title={Squeezed state generation in photonic crystal microcavities},
  author={Banaee, MG and Young, Jeff F},
  journal={Optics Express},
  volume={16},
  number={25},
  pages={20908--20919},
  year={2008},
  publisher={Optica Publishing Group},
  doi ={https://doi.org/10.1364/OE.16.020908}
}

@article{nadgaran2023squeezed,
  title={Squeezed states generation by nonlinear plasmonic waveguides: a novel analysis including loss, phase mismatch and source depletion},
  author={Nadgaran, Hamid and Izadi, Mohammad Amin and Nouroozi, Rahman},
  journal={Scientific Reports},
  volume={13},
  number={1},
  pages={1075},
  year={2023},
  publisher={Nature Publishing Group UK London},
  doi ={https://doi.org/10.1038/s41598-023-27949-x}
}

@article{colangelo2017simultaneous,
  title={Simultaneous tracking of spin angle and amplitude beyond classical limits},
  author={Colangelo, Giorgio and Ciurana, Ferran Martin and Bianchet, Lorena C and Sewell, Robert J and Mitchell, Morgan W},
  journal={Nature},
  volume={543},
  number={7646},
  pages={525--528},
  year={2017},
  publisher={Nature Publishing Group UK London},
  doi ={https://doi.org/10.1038/nature21434}
}

@article{hosten2016measurement,
  title={Measurement noise 100 times lower than the quantum-projection limit using entangled atoms},
  author={Hosten, Onur and Engelsen, Nils J and Krishnakumar, Rajiv and Kasevich, Mark A},
  journal={Nature},
  volume={529},
  number={7587},
  pages={505--508},
  year={2016},
  publisher={Nature Publishing Group UK London},
  doi = {https://doi.org/10.1038/nature16176}
}

@article{Stefszky_2011,
title = {An investigation of doubly-resonant optical parametric oscillators and nonlinear crystals for squeezing},
journal = {Journal of Physics B: Atomic, Molecular and Optical Physics},
author = {Michael Stefszky and Conor M Mow-Lowry and Kirk McKenzie and Sheon Chua and Ben C Buchler and Thomas Symul and David E McClelland and Ping Koy Lam},
year = {2010},
month = {dec},
publisher = {},
volume = {44},
number = {1},
pages = {015502},
doi ={10.1088/0264-9381/29/14/145015}
}

@article{chen2011conditional,
  title={Conditional spin squeezing of a large ensemble via the vacuum Rabi splitting},
  author={Chen, Zilong and Bohnet, Justin G and Sankar, Shannon R and Dai, Jiayan and Thompson, James K},
  journal={Physical review letters},
  volume={106},
  number={13},
  pages={133601},
  year={2011},
  publisher={APS},
  doi={https://doi.org/10.1103/PhysRevLett.106.133601}
}

@article{schleier2010states,
  title={States of an ensemble of two-level atoms with reduced quantum uncertainty},
  author={Schleier-Smith, Monika H and Leroux, Ian D and Vuleti{\'c}, Vladan},
  journal={Physical review letters},
  volume={104},
  number={7},
  pages={073604},
  year={2010},
  publisher={APS},
  doi ={https://doi.org/10.1103/PhysRevLett.104.073604}
}

@article{leroux2010implementation,
  title={Implementation of cavity squeezing of a collective atomic spin},
  author={Leroux, Ian D and Schleier-Smith, Monika H and Vuleti{\'c}, Vladan},
  journal={Physical Review Letters},
  volume={104},
  number={7},
  pages={073602},
  year={2010},
  publisher={APS},
  doi ={https://doi.org/10.1103/PhysRevLett.104.073602}
}

@article{xin2023long,
  title={Long-Lived Squeezed Ground States in a Quantum Spin Ensemble},
  author={Xin, Lin and Barrios, Maryrose and Cohen, Julia T and Chapman, Michael S},
  journal={Physical Review Letters},
  volume={131},
  number={13},
  pages={133402},
  year={2023},
  publisher={APS},
  doi={https://doi.org/10.1103/PhysRevLett.131.133402}
}

@article{berrada2013integrated,
  title={Integrated {M}ach--{Z}ehnder interferometer for {B}ose--{E}instein condensates},
  author={Berrada, Tarik and Van Frank, Sandrine and B{\"u}cker, Robert and Schumm, Thorsten and Schaff, J-F and Schmiedmayer, J{\"o}rg},
  journal={Nature communications},
  volume={4},
  number={1},
  pages={2077},
  year={2013},
  publisher={Nature Publishing Group UK London},
  doi ={https://doi.org/10.1038/ncomms3077}
}

@article{strobel2014fisher,
  title={Fisher information and entanglement of non-Gaussian spin states},
  author={Strobel, Helmut and Muessel, Wolfgang and Linnemann, Daniel and Zibold, Tilman and Hume, David B and Pezz{\`e}, Luca and Smerzi, Augusto and Oberthaler, Markus K},
  journal={Science},
  volume={345},
  number={6195},
  pages={424--427},
  year={2014},
  publisher={American Association for the Advancement of Science},
 doi={10.1126/science.1250147}
}

@article{lucke2014detecting,
  title={Detecting multiparticle entanglement of {D}icke states},
  author={L{\"u}cke, Bernd and Peise, Jan and Vitagliano, Giuseppe and Arlt, Jan and Santos, Luis and T{\'o}th, G{\'e}za and Klempt, Carsten},
  journal={Physical review letters},
  volume={112},
  number={15},
  pages={155304},
  year={2014},
  publisher={APS},
 doi ={https://doi.org/10.1103/PhysRevLett.112.155304}
}

@article{von2022parity,
  title={Parity measurement in the strong dispersive regime of circuit quantum acoustodynamics},
  author={Von L{\"u}pke, Uwe and Yang, Yu and Bild, Marius and Michaud, Laurent and Fadel, Matteo and Chu, Yiwen},
  journal={Nature Physics},
  volume={18},
  number={7},
  pages={794--799},
  year={2022},
  publisher={Nature Publishing Group UK London},
  doi ={https://doi.org/10.1038/s41567-022-01591-2}
}

@PREAMBLE{
 "\providecommand{\noopsort}[1]{}" 
 # "\providecommand{\singleletter}[1]{#1}%" 
}

@book{byrnes2021quantum,
  title={Quantum atom optics: Theory and applications to quantum technology},
  author={Byrnes, Tim and Ilo-Okeke, Ebubechukwu O},
  year={2021},
  publisher={Cambridge university press},
  doi={https://doi.org/10.1017/9781108975353}
}

@article{kondappan2023imaginary,
  title={Imaginary-time evolution with quantum nondemolition measurements: Multiqubit interactions via measurement nonlinearities},
  author={Kondappan, Manikandan and Chaudhary, Manish and Ilo-Okeke, Ebubechukwu O and Ivannikov, Valentin and Byrnes, Tim},
  journal={Physical Review A},
  volume={107},
  number={4},
  pages={042616},
  year={2023},
  publisher={APS},
  doi={https://doi.org/10.1103/PhysRevA.107.042616}
}

@article{giovannetti2011advances,
  title={Advances in quantum metrology},
  author={Giovannetti, Vittorio and Lloyd, Seth and Maccone, Lorenzo},
  journal={Nature photonics},
  volume={5},
  number={4},
  pages={222--229},
  year={2011},
  publisher={Nature Publishing Group UK London},
  doi={https://doi.org/10.1038/nphoton.2011.35}
}

@article{braunstein2005quantum,
  title = {Quantum information with continuous variables},
  author = {Braunstein, Samuel L. and van Loock, Peter},
  journal = {Rev. Mod. Phys.},
  volume = {77},
  issue = {2},
  pages = {513--577},
  numpages = {0},
  year = {2005},
  month = {Jun},
  publisher = {American Physical Society},
  doi = {10.1103/RevModPhys.77.513},
  url = {https://link.aps.org/doi/10.1103/RevModPhys.77.513}
}

@article{mao2023measurement,
  title = {Measurement-Based Deterministic Imaginary Time Evolution},
  author = {Mao, Yuping and Chaudhary, Manish and Kondappan, Manikandan and Shi, Junheng and Ilo-Okeke, Ebubechukwu O. and Ivannikov, Valentin and Byrnes, Tim},
  journal = {Phys. Rev. Lett.},
  volume = {131},
  issue = {11},
  pages = {110602},
  numpages = {6},
  year = {2023},
  month = {Sep},
  publisher = {American Physical Society},
  doi = {10.1103/PhysRevLett.131.110602},
  url = {https://link.aps.org/doi/10.1103/PhysRevLett.131.110602}
}

@article{LutterbachPRL97,
  title = {Method for Direct Measurement of the Wigner Function in Cavity QED and Ion Traps},
  author = {Lutterbach, L. G. and Davidovich, L.},
  journal = {Phys. Rev. Lett.},
  volume = {78},
  issue = {13},
  pages = {2547--2550},
  numpages = {0},
  year = {1997},
  month = {Mar},
  publisher = {American Physical Society},
  doi = {10.1103/PhysRevLett.78.2547},
  url = {https://link.aps.org/doi/10.1103/PhysRevLett.78.2547}
}

@article{LeibfriedRMP03,
  title = {Quantum dynamics of single trapped ions},
  author = {Leibfried, D. and Blatt, R. and Monroe, C. and Wineland, D.},
  journal = {Rev. Mod. Phys.},
  volume = {75},
  issue = {1},
  pages = {281--324},
  numpages = {0},
  year = {2003},
  month = {Mar},
  publisher = {American Physical Society},
  doi = {10.1103/RevModPhys.75.281},
  url = {https://link.aps.org/doi/10.1103/RevModPhys.75.281}
}

@article{Vlastakis2013,
author={Vlastakis, Brian
and Kirchmair, Gerhard
and Leghtas, Zaki
and Nigg, Simon E.
and Frunzio, Luigi
and Girvin, S. M.
and Mirrahimi, Mazyar
and Devoret, M. H.
and Schoelkopf, R. J.},
title={Deterministically Encoding Quantum Information Using 100-Photon {S}chr{\"o}dinger Cat States},
journal={Science},
year={2013},
month={Nov},
day={01},
publisher={American Association for the Advancement of Science},
volume={342},
number={6158},
pages={607-610},
doi={10.1126/science.1243289},
url={https://doi.org/10.1126/science.1243289}
}

@book{gerry2023introductory,
  title={Introductory quantum optics},
  author={Gerry, Christopher C and Knight, Peter L},
  year={2023},
  publisher={Cambridge university press},
  doi ={https://doi.org/10.1017/CBO9780511791239}
}

@article{royer1977wigner,
  title={Wigner function as the expectation value of a parity operator},
  author={Royer, Antoine},
  journal={Physical Review A},
  volume={15},
  number={2},
  pages={449},
  year={1977},
  publisher={APS},
  doi={https://doi.org/10.1103/PhysRevA.15.449}
}

@article{vahlbruch2016detection,
  title={Detection of 15 d{B} squeezed states of light and their application for the absolute calibration of photoelectric quantum efficiency},
  author={Vahlbruch, Henning and Mehmet, Moritz and Danzmann, Karsten and Schnabel, Roman},
  journal={Physical review letters},
  volume={117},
  number={11},
  pages={110801},
  year={2016},
  publisher={APS},
  doi    = {https://doi.org/10.1103/PhysRevLett.117.110801}
}

@article{byrnes2011accelerated,
  title={Accelerated optimization problem search using {B}ose--{E}instein condensation},
  author={Byrnes, Tim and Yan, Kai and Yamamoto, Yoshihisa},
  journal={New Journal of Physics},
  volume={13},
  number={11},
  pages={113025},
  year={2011},
  publisher={IOP Publishing},
  doi={10.1088/1367-2630/13/11/113025}
}

@article{scully2009super,
  title={The super of superradiance},
  author={Scully, Marlan O and Svidzinsky, Anatoly A},
  journal={Science},
  volume={325},
  number={5947},
  pages={1510--1511},
  year={2009},
  publisher={American Association for the Advancement of Science},
  doi={10.1126/science.1176695}
}

@article{heeres2017implementing,
  title={Implementing a universal gate set on a logical qubit encoded in an oscillator},
  author={Heeres, Reinier W and Reinhold, Philip and Ofek, Nissim and Frunzio, Luigi and Jiang, Liang and Devoret, Michel H and Schoelkopf, Robert J},
  journal={Nature communications},
  volume={8},
  number={1},
  pages={94},
  year={2017},
  publisher={Nature Publishing Group UK London},
  doi={https://doi.org/10.1038/s41467-017-00045-1}
}

@article{gottesman2001encoding,
  title={Encoding a qubit in an oscillator},
  author={Gottesman, Daniel and Kitaev, Alexei and Preskill, John},
  journal={Physical Review A},
  volume={64},
  number={1},
  pages={012310},
  year={2001},
  publisher={APS},
  doi={https://doi.org/10.1103/PhysRevA.64.012310}
}

@article{weigand2018generating,
  title={Generating grid states from {S}chr{\"o}dinger-cat states without postselection},
  author={Weigand, Daniel J and Terhal, Barbara M},
  journal={Physical Review A},
  volume={97},
  number={2},
  pages={022341},
  year={2018},
  publisher={APS},
  doi={https://doi.org/10.1103/PhysRevA.97.022341}
}

@article{rozpkedek2021quantum,
  title={Quantum repeaters based on concatenated bosonic and discrete-variable quantum codes},
  author={Rozpedek, Filip and Noh, Kyungjoo and Xu, Qian and Guha, Saikat and Jiang, Liang},
  journal={npj Quantum Information},
  volume={7},
  number={1},
  pages={102},
  year={2021},
  publisher={Nature Publishing Group UK London},
  doi={https://doi.org/10.1038/s41534-021-00438-7}
}

@article{valahu2025quantum,
  title={Quantum-enhanced multiparameter sensing in a single mode},
  author={Valahu, Christophe H and Stafford, Matthew P and Huang, Zixin and Matsos, Vassili G and Millican, Maverick J and Chalermpusitarak, Teerawat and Menicucci, Nicolas C and Combes, Joshua and Baragiola, Ben Q and Tan, Ting Rei},
  journal={Science Advances},
  volume={11},
  number={39},
  pages={eadw9757},
  year={2025},
  publisher={American Association for the Advancement of Science},
  doi={10.1126/sciadv.adw9757}
}

@article{royer2022encoding,
  title={Encoding qubits in multimode grid states},
  author={Royer, Baptiste and Singh, Shraddha and Girvin, Steven M},
  journal={PRX Quantum},
  volume={3},
  number={1},
  pages={010335},
  year={2022},
  publisher={APS},
  doi={https://doi.org/10.1103/PRXQuantum.3.010335}
}

@article{brady2024advances,
  title={Advances in bosonic quantum error correction with {G}ottesman--{K}itaev--{P}reskill codes: Theory, engineering and applications},
  author={Brady, Anthony J and Eickbusch, Alec and Singh, Shraddha and Wu, Jing and Zhuang, Quntao},
  journal={Progress in Quantum Electronics},
  volume={93},
  pages={100496},
  year={2024},
  publisher={Elsevier},
  doi={https://doi.org/10.1016/j.pquantelec.2023.100496}
}

@article{fluhmann2019encoding,
  title={Encoding a qubit in a trapped-ion mechanical oscillator},
  author={Fl{\"u}hmann, Christa and Nguyen, Thanh Long and Marinelli, Matteo and Negnevitsky, Vlad and Mehta, Karan and Home, JP},
  journal={Nature},
  volume={566},
  number={7745},
  pages={513--517},
  year={2019},
  publisher={Nature Publishing Group UK London},
  doi={https://doi.org/10.1038/s41586-019-0960-6}
}

@article{campagne2020quantum,
  title={Quantum error correction of a qubit encoded in grid states of an oscillator},
  author={Campagne-Ibarcq, Philippe and Eickbusch, Alec and Touzard, Steven and Zalys-Geller, Evan and Frattini, Nicholas E and Sivak, Volodymyr V and Reinhold, Philip and Puri, Shruti and Shankar, Shyam and Schoelkopf, Robert J and others},
  journal={Nature},
  volume={584},
  number={7821},
  pages={368--372},
  year={2020},
  publisher={Nature Publishing Group UK London},
  doi={https://doi.org/10.1038/s41586-020-2603-3}
}

@article{konno2024logical,
  title={Logical states for fault-tolerant quantum computation with propagating light},
  author={Konno, Shunya and Asavanant, Warit and Hanamura, Fumiya and Nagayoshi, Hironari and Fukui, Kosuke and Sakaguchi, Atsushi and Ide, Ryuhoh and China, Fumihiro and Yabuno, Masahiro and Miki, Shigehito and others},
  journal={Science},
  volume={383},
  number={6680},
  pages={289--293},
  year={2024},
  publisher={American Association for the Advancement of Science},
  doi ={10.1126/science.adk7560}
}

@article{larsen2025integrated,
  title={Integrated photonic source of {G}ottesman--{K}itaev--{P}reskill qubits},
  author={Larsen, MV and Bourassa, JE and Kocsis, S and Tasker, JF and Chadwick, RS and Gonz{\'a}lez-Arciniegas, C and Hastrup, J and Lopetegui-Gonz{\'a}lez, CE and Miatto, FM and Motamedi, A and others},
  journal={Nature},
  pages={1--5},
  year={2025},
  publisher={Nature Publishing Group UK London},
  doi ={https://doi.org/10.1038/s41586-025-09044-5}
}

@article{feng2025quantum,
  title={Quantum teleportation of cat states with binary-outcome measurements},
  author={Feng, Jingyan and Zhang, Mohan and Fadel, Matteo and Byrnes, Tim},
  journal={Quantum Science and Technology},
  volume={10},
  number={4},
  pages={045022},
  year={2025},
  publisher={IOP Publishing},
  doi={10.1088/2058-9565/adf572}
}

@article{lee2025photonic,
  title={Photonic Hybrid Quantum Computing},
  author={Lee, Jaehak and Omkar, Srikrishna and Teo, Yong Siah and Lee, Seok-Hyung and Kwon, Hyukjoon and Kim, MS and Jeong, Hyunseok},
  journal={arXiv preprint arXiv:2510.00534},
  year={2025},
  doi={https://doi.org/10.48550/arXiv.2510.00534}

}

@article{chen2024efficient,
  title={Efficient preparation of the {AKLT} state with measurement-based imaginary time evolution},
  author={Chen, Tianqi and Byrnes, Tim},
  journal={Quantum},
  volume={8},
  pages={1557},
  year={2024},
  publisher={Verein zur F{\"o}rderung des Open Access Publizierens in den Quantenwissenschaften},
  doi={https://doi.org/10.22331/q-2024-12-10-1557}
}

@article{pezze2018quantum,
  title={Quantum metrology with nonclassical states of atomic ensembles},
  author={Pezze, Luca and Smerzi, Augusto and Oberthaler, Markus K and Schmied, Roman and Treutlein, Philipp},
  journal={Reviews of Modern Physics},
  volume={90},
  number={3},
  pages={035005},
  year={2018},
  publisher={APS},
  doi ={https://doi.org/10.1103/RevModPhys.90.035005}
}

@article{HamiltonPRL17,
  title = {Gaussian Boson Sampling},
  author = {Hamilton, Craig S. and Kruse, Regina and Sansoni, Linda and Barkhofen, Sonja and Silberhorn, Christine and Jex, Igor},
  journal = {Phys. Rev. Lett.},
  volume = {119},
  issue = {17},
  pages = {170501},
  numpages = {5},
  year = {2017},
  month = {Oct},
  publisher = {American Physical Society},
  doi = {10.1103/PhysRevLett.119.170501},
  url = {https://link.aps.org/doi/10.1103/PhysRevLett.119.170501}
}

@article{LundPRL14,
  title = {Boson Sampling from a Gaussian State},
  author = {Lund, A. P. and Laing, A. and Rahimi-Keshari, S. and Rudolph, T. and O'Brien, J. L. and Ralph, T. C.},
  journal = {Phys. Rev. Lett.},
  volume = {113},
  issue = {10},
  pages = {100502},
  numpages = {5},
  year = {2014},
  month = {Sep},
  publisher = {American Physical Society},
  doi = {10.1103/PhysRevLett.113.100502},
  url = {https://link.aps.org/doi/10.1103/PhysRevLett.113.100502}
}

@article{vanner2011pulsed,
  title={Pulsed quantum optomechanics},
  author={Vanner, Michael R and Pikovski, Igor and Cole, Garrett D and Kim, Myung Shik and Brukner, {\v{C}} and Hammerer, Klemens and Milburn, Gerard J and Aspelmeyer, Markus},
  journal={Proceedings of the National Academy of Sciences},
  volume={108},
  number={39},
  pages={16182--16187},
  year={2011},
  publisher={National Academy of Sciences},
  doi = {https://doi.org/10.1073/pnas.1105098108}
}

@article{szorkovszky2011mechanical,
  title={Mechanical squeezing via parametric amplification and weak measurement},
  author={Szorkovszky, Alex and Doherty, Andrew C and Harris, Glen I and Bowen, Warwick P},
  journal={Physical review letters},
  volume={107},
  number={21},
  pages={213603},
  year={2011},
  publisher={APS},
  doi = {https://doi.org/10.1103/PhysRevLett.107.213603}
}

@article{wollman2015quantum,
  title={Quantum squeezing of motion in a mechanical resonator},
  author={Wollman, Emma Edwina and Lei, CU and Weinstein, AJ and Suh, J and Kronwald, A and Marquardt, F and Clerk, Aashish A and Schwab, KC},
  journal={Science},
  volume={349},
  number={6251},
  pages={952--955},
  year={2015},
  publisher={American Association for the Advancement of Science},
  doi = {DOI: 10.1126/science.aac5138}
}

@article{pan2023protecting,
  title={Protecting the quantum interference of cat states by phase-space compression},
  author={Pan, Xiaozhou and Schwinger, Jonathan and Huang, Ni-Ni and Song, Pengtao and Chua, Weipin and Hanamura, Fumiya and Joshi, Atharv and Valadares, Fernando and Filip, Radim and Gao, Yvonne Y},
  journal={Physical Review X},
  volume={13},
  number={2},
  pages={021004},
  year={2023},
  publisher={APS},
  doi= {https://doi.org/10.1103/PhysRevX.13.021004}
}

@article{krisnanda2026direct,
  title={Direct estimation of arbitrary observables of an oscillator},
  author={Krisnanda, Tanjung and Valadares, Fernando and Chu, Kyle Timothy Ng and Song, Pengtao and Copetudo, Adrian and Fontaine, Clara Yun and Lachman, Luka{\'a}{\v{s}} and Filip, Radim and Gao, Yvonne Y},
  journal={Physical Review Research},
  volume={8},
  number={2},
  pages={023026},
  year={2026},
  publisher={APS},
  doi= {DOI: https://doi.org/10.1103/72f2-tgwp}
}

@article{heeres2015snap,
  title={Cavity State Manipulation Using Photon-Number Selective Phase Gates},
  author={Heeres, Reinier W. and Vlastakis, Brian and Holland, Eric and
          Krastanov, Stefan and Albert, Victor V. and Frunzio, Luigi and
          Jiang, Liang and Schoelkopf, Robert J.},
  journal={Physical Review Letters},
  volume={115},
  number={13},
  pages={137002},
  year={2015},
  publisher={APS}
}

@article{khaneja2005optimal,
  title={Optimal control of coupled spin dynamics: design of NMR pulse sequences by gradient ascent algorithms},
  author={Khaneja, Navin and Reiss, Timo and Kehlet, Cindie and Schulte-Herbr{\"u}ggen, Thomas and Glaser, Steffen J},
  journal={Journal of magnetic resonance},
  volume={172},
  number={2},
  pages={296--305},
  year={2005},
  publisher={Elsevier}
}

@article{deleglise2008reconstruction,
  title={Reconstruction of non-classical cavity field states with snapshots of their decoherence},
  author={Deleglise, Samuel and Dotsenko, Igor and Sayrin, Clement and Bernu, Julien and Brune, Michel and Raimond, Jean-Michel and Haroche, Serge},
  journal={Nature},
  volume={455},
  number={7212},
  pages={510--514},
  year={2008},
  publisher={Nature Publishing Group UK London}
}

@article{vlastakis2013deterministically,
  title={Deterministically encoding quantum information using 100-photon Schr{\"o}dinger cat states},
  author={Vlastakis, Brian and Kirchmair, Gerhard and Leghtas, Zaki and Nigg, Simon E and Frunzio, Luigi and Girvin, Steven M and Mirrahimi, Mazyar and Devoret, Michel H and Schoelkopf, Robert J},
  journal={Science},
  volume={342},
  number={6158},
  pages={607--610},
  year={2013},
  publisher={American Association for the Advancement of Science}
}

@article{wang2016schrodinger,
  title={A Schr{\"o}dinger cat living in two boxes},
  author={Wang, Chen and Gao, Yvonne Y and Reinhold, Philip and Heeres, Reinier W and Ofek, Nissim and Chou, Kevin and Axline, Christopher and Reagor, Matthew and Blumoff, Jacob and Sliwa, KM and others},
  journal={Science},
  volume={352},
  number={6289},
  pages={1087--1091},
  year={2016},
  publisher={American Association for the Advancement of Science}
}

@article{brune1992manipulation,
  title={Manipulation of photons in a cavity by dispersive atom-field coupling: Quantum-nondemolition measurements and generation of ‘‘Schr{\"o}dinger cat’’states},
  author={Brune, Michel and Haroche, Serge and Raimond, Jean-Michel and Davidovich, Luis and Zagury, Nicim},
  journal={Physical Review A},
  volume={45},
  number={7},
  pages={5193},
  year={1992},
  publisher={APS}
}

@article{lloyd1999quantum,
  title={Quantum computation over continuous variables},
  author={Lloyd, Seth and Braunstein, Samuel L},
  journal={Physical Review Letters},
  volume={82},
  number={8},
  pages={1784},
  year={1999},
  publisher={APS}
}
%\bibliography{C:/Users/Ebube/Documents/Projects/GeometricPhaseGate/v13Revised/paperrefs}
% 2) Copy the .bbl file to below and comment out the above two lines. 
%\bibliography{C:/Users/win/Documents/clusterQND/paper}


\end{document}